\documentclass[10pt,a4paper]{article}

\usepackage{amsmath,amsgen,latexsym}
\usepackage{amstext,amssymb,amsfonts,latexsym}
\usepackage{theorem}
\usepackage{pifont}
\usepackage{graphicx}
\usepackage[hyphens]{url}

 \newcommand{\bs}{\bigskip}
 \newcommand{\ms}{\medskip}
 \newcommand{\n}{\noindent}
 \newcommand{\s}{\smallskip}
 \newcommand{\hs}[1]{\hspace*{ #1 mm}}
 \newcommand{\vs}[1]{\vspace*{ #1 mm}}

 \newcommand{\setempty}{\varnothing}
 
 \newcommand{\nat}{\mathbb{N}}
 
 \newcommand{\integer}{\mathbb{Z}}

 \newcommand{\co}{\mathrm{co}\mbox{-}}

 \newcommand{\etalc}{\textrm{et al.}}

 \newcommand{\AAA}{{\cal A}}
 \newcommand{\BB}{{\cal B}}
 
 \newcommand{\CC}{{\cal C}}
 \newcommand{\FF}{{\cal F}}
 \newcommand{\DD}{{\cal D}}

 \newcommand{\KK}{{\cal K}}
 \newcommand{\LL}{{\cal L}}
 \newcommand{\NN}{{\cal N}}
 
 \newcommand{\MM}{{\cal M}}

 \newcommand{\PP}{{\cal P}}

 \newcommand{\dl}{\mathrm{L}}
 \newcommand{\nl}{\mathrm{NL}}
 \newcommand{\ul}{\mathrm{UL}}
 \newcommand{\p}{\mathrm{P}}
 \newcommand{\np}{\mathrm{NP}}

 \newcommand{\poly}{\mathrm{poly}}

\theoremstyle{plain}
\theoremheaderfont{\bfseries}
 \newtheorem{theorem}{Theorem}[section]
 \newtheorem{lemma}[theorem]{Lemma}
 \newtheorem{proposition}[theorem]{{\bf Proposition}}
 \newtheorem{corollary}[theorem]{Corollary}

 {\theorembodyfont{\rmfamily}
  }
 {\theorembodyfont{\rmfamily} }
 {\theorembodyfont{\rmfamily} }

 \newenvironment{proofof}[1]{\vspace*{5mm} \par \noindent
         {\bf Proof of #1.\hs{2}}}{\hfill$\Box$ \vspace*{3mm}}

 \newtheorem{yclaim}[theorem]{Claim}

 \newenvironment{proof}{\par \noindent
            {\bf Proof. \hs{2}}}{\hfill$\Box$ \vspace*{3mm}}

\newcommand{\ignore}[1]{}

\newcommand{\track}[2]{[\:\begin{subarray}{c} #1 \\%
      #2 \end{subarray} ]}

 \newcommand{\oned}{1\mathrm{D}}
 
 \newcommand{\onebq}{1\mathrm{BQ}}

 \newcommand{\onen}{1\mathrm{N}}
 \newcommand{\twod}{2\mathrm{D}}
 
 \newcommand{\twon}{2\mathrm{N}}

 \newcommand{\para}{\mathrm{para}\mbox{-}}

 \newcommand{\onecequal}{1\mathrm{C_{=}}}

 \newcommand{\oneu}{1\mathrm{U}}

 \newcommand{\onefewu}{1\mathrm{FewU}}
 \newcommand{\onefew}{1\mathrm{Few}}
 \newcommand{\onereachu}{1\mathrm{ReachU}}
 \newcommand{\onereachfew}{1\mathrm{ReachFew}}
 \newcommand{\onereachfewu}{1\mathrm{ReachFewU}}
 \newcommand{\twou}{2\mathrm{U}}
 \newcommand{\twofew}{2\mathrm{Few}}
 \newcommand{\tworeachu}{2\mathrm{ReachU}}
 \newcommand{\twofewu}{2\mathrm{FewU}}
 \newcommand{\tworeachfew}{2\mathrm{ReachFew}}
 
 \newcommand{\tworeachfewu}{2\mathrm{ReachFewU}}

 \newcommand{\twodpd}{2\mathrm{DPD}}

 \newcommand{\up}{\mathrm{UP}}
 \newcommand{\fewp}{\mathrm{FewP}}
 \newcommand{\fewl}{\mathrm{FewL}}
 
 \newcommand{\slog}{\mathrm{log}}
 \newcommand{\sexp}{\mathrm{exp}}
 \newcommand{\supexp}{\mathrm{supexp}}
 \newcommand{\PHSP}{\mathrm{PHSP}}

 \newcommand{\dbraleft}{\:[\!\!\![\;}
 \newcommand{\dbraright}{\;]\!\!\!]\,}

\begin{document}

\pagestyle{plain}
\setcounter{page}{1}

\begin{center}
{\Large {\bf Unambiguity and Fewness for Nonuniform Families of Polynomial-Size Nondeterministic Finite  Automata}}\footnote{This work    corrects and also significantly alters the preliminary report that appeared in the Proceedings of the 16th International Conference on Reachability Problems (RP 2022), Kaiserslautern, Germany, October 17--21, 2022, Lecture Notes in Computer Science, vol. 13608, pp. 77--92, Springer Cham, 2022.}
\bs\\
{\sc Tomoyuki Yamakami}\footnote{Present Affiliation: Faculty of Engineering, University of Fukui, 3-9-1 Bunkyo, Fukui 910-8507, Japan}
\bs\\
\end{center}


\begin{abstract}
Nonuniform families of polynomial-size finite automata, which are series of indexed finite automata having polynomially many inner states, are used in the past literature to solve nonuniform families of promise decision problems. Among such nonuniform families of finite automata, we focus our attention, in particular, on the variants of
nondeterministic finite automata, which have at most ``one'' (unambiguous), ``polynomially many'' (few) accepting computation paths, or unambiguous/few computation paths leading to each fixed configuration.  When such machines are limited to make only one-way head moves, we can prove  with no unproven hardness assumptions that some of these variants are different in computational power from each other. As for two-way machines restricted to instances of polynomially-bounded length, families of two-way polynomial-size nondeterministic finite automata are equivalent in power to families of polynomial-size unambiguous finite automata.

\s
\n{\bf Key words}
nonuniform state complexity, finite automata, accepting computation path, unambiguous, fewness
\end{abstract}

\sloppy
\section{Historical Background and Quick Overview}\label{sec:introduction}

As an introduction, we briefly go over its historical background and then take a quick overview of the major contributions of this work.

\subsection{Unambiguity and Fewness in Complexity Theory}\label{sec:L-vs-NL}

The number of accepting computation paths of an underlying nondeterministic machine has been a centerpiece of intensive research over the decades because the acceptance criteria of how the machine ``accepts'' each instance are of great importance for nondeterministic computation and this is indeed a key to the full understandings of nondeterministic computation.

Among various acceptance criteria, the notion of \emph{unambiguity} for an underlying nondeterministic machine, which means at most one accepting computation path, has drawn a great interest.
The study of unambiguous context-free languages, for example, is one of the important subjects in formal language theory because there exist practical parsing algorithms for those languages.

In computational complexity theory, the unambiguity issues have been discussed since Valiant \cite{Val76} introduced the unambiguous polynomial-time complexity class, known as $\up$, in connection to the existence of one-way functions, which play an essential role in modern cryptography. When we allow more than one accepting computation paths but limited to ``few'' (i.e., polynomially many) paths, we then obtain the complexity class called $\fewp$, introduced in \cite{All86,AR88}. It follows that $\p\subseteq \up \subseteq \fewp \subseteq \np$; however, it still remains open whether or not these inclusions are proper.
A series of papers in the past literature further proposed various refinements of unambiguous languages and that has enriched the study of computational complexity of languages.
Unarguably, the study on the behaviors of nondeterministic machines producing a various number of accepting computation paths is a key in computational complexity theory to the anticipated separation between $\p$ and $\np$.

For the model of space-bounded machines, on the contrary, the logarithmic-space (or log-space, for short) analogues of $\up$ and $\fewp$, denoted $\ul$ and $\fewl$,  were discussed in the  1990s \cite{BJLR91} using nondeterministic logarithmic-space computations in comparison with the polynomial-time computations.
The memory space restriction sometimes presents a quite different landscape from the runtime restriction.
For instance, whereas $\np$ and $\co\np$ are not known to coincide, their log-space counterparts, $\nl$ and $\co\nl$, are in fact equal \cite{Imm88,Sze88}.

With the help of Karp-Lipton style advice \cite{KL82}, Reinhardt and Allender \cite{RA00}, for instance,  managed to prove the equivalence between $\nl/\poly$ and $\ul/\poly$ although $\nl$ and $\ul$ themselves are still unknown to coincide, where advice is an  external information source, provided to an underlying Turing machine in parallel to a standard input in order to enhance its computational power. This equivalence is not yet known between $\np$ and $\up$.

Bourke, Tewari, and Vinodchandran \cite{BTV09} and lately Pavan, Tewari, and Vinodchandran \cite{PTV12} introduced another intriguing  refinement of the aforementioned complexity classes associated with unambiguity and fewness notions, and they exhibited among such refined complexity classes the existence of a rich structure situated in between $\dl$ and $\nl$. Their refined classes include: $\mathrm{ReachUL}$, $\mathrm{ReachFewL}$, $\mathrm{ReachLFew}$, and $\mathrm{FewUL}$.
It is also imperative to expand and explore the nature of unambiguity and fewness of accepting computation paths in other computational models.
As such a model, we consider nonuniform finite automata families.

\subsection{Nonuniform Families of Polynomial-Size Finite Automata}\label{sec:finite-automata}

Let us turn our attention to ``finite(-state) automata'', which are  one of the simplest models of computation. Those machines have been intensively studied since its early introduction but, only since the late 1970s, ``nonuniform families'' of those machines have drawn  our attention.
In analogy to Boolean circuit families,  Berman and Lingas \cite{BL77} and Sakoda and Sipser \cite{SS78} studied nonuniform families of finite automata, indexed by natural numbers, of polynomial state complexity. Such a finite automata family is succinctly said to have \emph{polynomial size}.
A series of papers \cite{Gef12,Kap09,Kap12,Kap14,KP15, Yam19a,Yam19d,Yam21,Yam22,Yam23b} have since then made crucial contributions to establishing a coherent theory over \emph{nonuniform polynomial state complexity} of families of ``languages'', more generally, \emph{promise decision problems}.
Nonuniform machine families are generally used as a vehicle to solve those families of promise problems in quite efficient ways.
In such a nonuniform setting, there are two parameters to take into consideration: machine's index $n$ and input length $|x|$.
The notion of nonuniformity,  ranging from advice-enhanced Turing machines to families of Boolean circuits, is as important as that of uniformity in computational complexity theory.

A family of finite automata can be viewed as an analogue of a family of Boolean circuits but it is quite different in the following key point: while each circuit in a circuit family takes only inputs of a fixed length, a finite automaton in an automata family
can take inputs of arbitrary length. This makes it possible for us to discuss subfamilies of an automata family by freely restricting the size of inputs, which is called a ``ceiling'' \cite{Yam22}.
By choosing different ceilings, we can discuss the computational complexity of a wide variety of nonuniform families of finite automata.

The collection of families of promise decision problems that are solvable by nonuniform families of polynomial-size two-way deterministic finite automata is particularly denoted $\twod$, and its nondeterministic variant is denoted $\twon$ \cite{SS78}.
These automata families in the above-mentioned literature have been studied in direct connection to logarithmic-space advised complexity classes, such as $\dl/\poly$ and $\nl/\poly$ \cite{BL77,Kap14}. More precisely, we denote by $\twod/\poly$ and $\twon/\poly$ the restrictions of $\twod$ and $\twon$ on only instances having polynomial ceilings. It was proven that the collapse of $\twon/\poly$ to $\twod/\poly$ occurs exactly when $\nl/\poly$ equals $\dl/\poly$.
A similar connection was observed in \cite{Yam19a} between families of restricted nondeterministic finite automata and the so-called \emph{linear space hypothesis}, proposed in \cite{Yam17a}.

When underlying finite automata are limited to make only one-way head moves, complexity classes of promise problem families present us a quite different landscape. For the one-way analogues of $\twod$ and $\twon$, denoted respectively by $\oned$ and $\onen$, it was proven that $\oned\neq\onen\neq\co\onen$ \cite{SS78}.

In fact, the nonuniform nature of automata families provides enormous flexibility to solving families of promise problems. Manifestation of this fact has been demonstrated in the field of automata theory for various machine types, including deterministic, nondeterministic, probabilistic, alternating, quantum automata and also pushdown automata in the literature.

It has been expected to further expand the scope of the study on nonuniform polynomial state complexity theory to other types of polynomial-size finite automata families.

\subsection{New Challenges and Main Contributions}\label{sec:main-contribution}

The \emph{theory of nonuniform polynomial state complexity}  has not been well developed in depth and in scope. It is therefore imperative to replenish this theory by cultivating and examining structural properties of underlying finite automata and their nonuniform families.
In this work, we intend to explore such structural properties of
We modify the existing models of nonuniform families of nondeterministic finite automata by requiring various restrictions on their computation paths. 
For this purpose, we wish to adapt various notions of \cite{PTV12} associated with unambiguity and fewness notions to fit into our setting of nonuniform families of finite automata. Our intension here is to replenish the theory by making new challenges in the topics of unambiguity and fewness for polynomial-size families of finite automata,
because unambiguity and fewness are in fact important notions in automata theory.
Therefore, we wish to explore these notions founded on nonuniform  families of nondeterministic finite automata. In particular, we attempt to follow the aforementioned work of Pavan \etalc~\cite{PTV12} in our setting of nonuniform state complexity theory. For this purpose, we wish to identify state complexity classes as done in the logarithmic-space setting of \cite{PTV12}.
We intend to study the computational complexity of families of polynomial-size nondeterministic finite automata that satisfy various conditions on accepting computation concerning unambiguity and fewness.

In a spirit similar to \cite{PTV12}, we will introduce six nonuniform polynomial state complexity classes between $\oned$ and $\onen$  in Section \ref{sec:one-way-head} and between $\twod$ and $\twon$ in Section \ref{sec:two-way-head}.
We will then demonstrate various relationships among those nonuniform  complexity classes and their complement classes.
In the case of one-way models, in particular, we will prove separations among some of these complexity classes \emph{with no unproven hardness assumptions}.
In the case of two-way models, on the contrary, we will demonstrate the collapses of complexity classes by exploiting a close relation between nonuniform families of polynomial-size finite automata having polynomial ceilings and logarithmic-space advised computation.
It was remarked in \cite{Yam23b} that there is a close connection between nonuniform families of polynomial-size finite automata and one-tape linear-time Turing machines, in particular, between promise problem families in $\onecequal$ and languages in $\mathrm{1\mbox{-}C}_{=}\mathrm{LIN/lin}$. We will adapt such a close connection and exploit it to prove the  desired separations of this work.
Our result is summarized and illustrated in Figure \ref{fig:class-hierarchy}. The detailed explanation of the complexity classes in the figure will be given in Sections \ref{sec:one-way-head} and \ref{sec:two-way-head}.


\begin{figure}[t]
\centering
\includegraphics*[height=5.4cm]{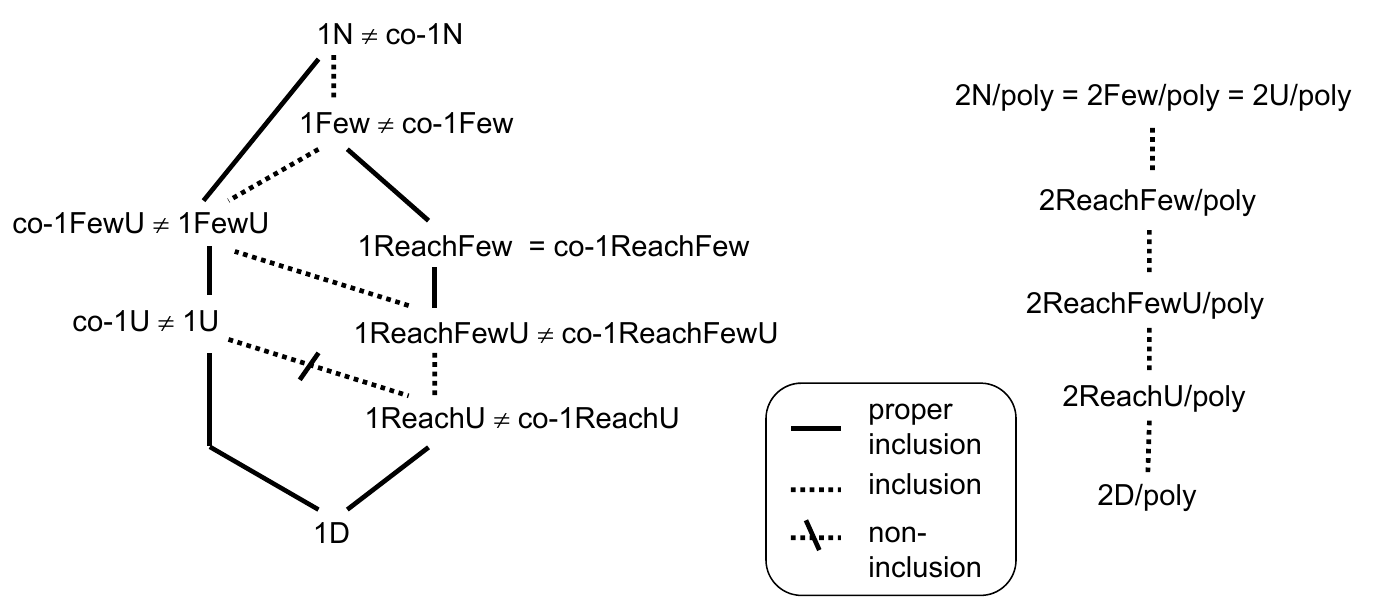}
\caption{Inclusion and collapse relations among families of promise problems discussed in this work except for $\onen\neq \co\onen$ \cite{SS78} and $\oned\neq \oneu\neq\co\oneu$ \cite{Yam23b}.}\label{fig:class-hierarchy}
\end{figure}


\section{Basic Notions and Notation}\label{sec:preparation}

We will explain the foundation of this work: finite automata models of polynomial state complexity.

\subsection{Numbers, Promise Problems, and Kolmogorov Complexity}\label{sec:number-promise}

Let $\nat$ (resp., $\nat^{+}$) denote the set of all nonnegative (resp., positive) integers. Notice that $\nat=\nat^{+}\cup\{0\}$. Given two integers $m$ and $n$ with $m\leq n$, $[m,n]_{\integer}$ denotes the \emph{integer interval} $\{m,m+1,m+2,\ldots,n\}$. When $n\in\nat^{+}$, $[1,n]_{\integer}$ is abbreviated as $[n]$ for simplicity.
Given a set $A$, the notation $\|A\|$ denotes the \emph{cardinality} of $A$ and $\PP(A)$ denotes the \emph{power set} of $A$.

A \emph{polynomial} in this paper is assumed to have nonnegative integer coefficients. An \emph{exponential} refers to a function of the form $2^{p(n)}$ for an appropriate polynomial $p$. Any \emph{logarithm} is assumed to take the base 2 and any \emph{logarithmic function} $f$ has the form $a\log x+b$ for certain constants $a,b\geq0$.
For convenience, we say that a function $f$ on $\nat$ (i.e., from $\nat$ to $\nat$) \emph{super-exponential} if, for any polynomial $p$, $f(n)>2^{p(n)}$ holds for all but finitely many numbers $n\in\nat$.
In contrast, a function $f$ is \emph{sub-exponential} if there exists a constant $\alpha\in[0,1)$ satisfying $f(n)\leq 2^{n^{\alpha}}$ for all but finitely many numbers $n\in\nat$.


An \emph{alphabet} is a finite nonempty set of ``symbols'' or ``letters''. Given such an alphabet $\Sigma$, a \emph{string} over $\Sigma$ refers to a finite sequence of symbols in $\Sigma$. The length of this sequence is called the \emph{length} of the string. In particular, the \emph{empty string} is a unique string of length $0$ and is denoted by $\varepsilon$.

We briefly explain (nonuniform) families of promise (decision) problems as described in \cite{Yam19a,Yam19d,Yam21,Yam22,Yam23b}.  Given an alphabet $\Sigma$, a \emph{promise (decision) problem over $\Sigma$} is a pair $(A,B)$ of sets satisfying $A\cup B \subseteq \Sigma^*$ and $A\cap B=\setempty$. Instances in $A$ are called \emph{positive} (or YES-instances) and instances in $B$ are  \emph{negative} (or NO-instances).
To clarify the use of positive/negative instances, we often use  the superscripts of $(+)$ and $(-)$ to express a promise problem as  $(L^{(+)},L^{(-)})$. We further consider a family $\LL$ of promise problems $(L_n^{(+)},L_n^{(-)})$ for all indices $n\in\nat$; however, we focus only on  families of promise problems over the ``same'' alphabet\footnote{In \cite{Kap09,Kap12,Kap14,KP15}, this extra condition is not required.}
$\Sigma$.
We consider a family $\LL=\{(L_n^{(+)},L_n^{(-)})\}_{n\in\nat}$ of those promise problems.
From this requirement, we often denote
the union  $L_n^{(+)}\cup L_n^{(-)}$ by $\Sigma_n$ whereas the notation $\Sigma^n$ is reserved for the set of strings of length $n$.
Any string in $\Sigma_n$ indicates  a \emph{valid} (or a \emph{promised}) instance over $\Sigma$.
The \emph{complement} of $\LL$, denoted $\co\LL$, consists
of $(L_n^{(-)},L_n^{(+)})$ for all indices $n\in\nat$.


Let $U$ denote any universal (deterministic) Turing machine taking binary input strings and eventually produces binary output strings. Given any binary strings $x$ and $y$, the \emph{conditional Kolmogorov complexity of $x$ conditional to $y$}, denoted $C(x|y)$, is the length of the shortest binary string $p$ such that $U$ on inputs $(p,y)$ produces $x$ on its output tape. We use the notation $\varepsilon$ to denote the \emph{empty string}. If $y$ is $\varepsilon$, we write $C(x)$ instead of $C(x|\varepsilon)$. This is referred to as the \emph{(unconditional) Kolmogorov complexity of $x$}. For more detail, the reader refers to a textbook, e.g., \cite{LV08}.

\subsection{Families of Finite Automata}\label{sec:FA-CS}

Throughout this work, we mostly deal with a computational model of nondeterministic finite automata. For convenience, we freely abbreviate a \emph{one-way nondeterministic finite(-state) automaton} as a 1nfa and
a \emph{two-way nondeterministic finite(-state) automaton} as a 2nfa. Similarly, we call deterministic variants of 1nfa and 2nfa by 1dfa and 2dfa, respectively.

Our main target is ``families'' of finite automata of the same machine type over the same alphabet $\Sigma$. In particular, we use families of 2nfa's (as well as  2dfa's, 1nfa's, and 1dfa's), indexed by natural numbers, as a base model to solve families of promise problems.

Each 2nfa $M_n$ in such a family $\MM=\{M_n\}_{n\in\nat}$ is equipped with a semi-infinite input tape, whose tape cells are indexed by natural numbers starting with $0$. Formally, $M_n$ is expressed as
a septuple $(Q_n,\Sigma,{\{\rhd,\lhd\}},\delta_n, q_{0,n},Q_{acc,n},Q_{rej,n})$ with two designated endmarkers $\rhd$ and $\lhd$ and two sets $Q_{acc,n}$ and $Q_{rej,n}$ of accepting (inner) states and rejecting (inner) states satisfying both $Q_{acc,n}\cup Q_{rej,n}\subseteq Q_n$ and $Q_{acc,n}\cap Q_{rej,n}=\setempty$. Any inner state in $Q_{acc,n}\cup Q_{rej,n}$ are simply called a \emph{halting (inner) state}. The transition function $\delta$ is defined only on non-halting states.
with a finite set $Q$ of inner states, an alphabet $\Sigma$ with two designated  endmarkers $\rhd$ and $\lhd$, a transition function $\delta:(Q-Q_{halt})\times \check{\Sigma} \to \PP(Q\times D)$, the initial inner state $q_0$ ($\in Q$), and two sets $Q_{acc}$ and $Q_{rej}$ of accepting states and rejecting states with $Q_{acc}\cup Q_{rej}\subseteq Q$ and $Q_{acc}\cap Q_{rej}=\setempty$, where $Q_{halt} = Q_{acc}\cup Q_{rej}$, $\check{\Sigma} = \Sigma\cup\{{\rhd, \lhd}\}$, and $D=\{-1,0,+1\}$. Any transition $(p,d)\in\delta(q,\sigma)$ indicates that, in a single step, $M$ scans $\sigma$, changes its inner state from $q$ to $p$, and move its tape head in direction $d$ (i.e., $-1$ means ``to the left'', $+1$ means ``to the right'', and $0$ means ``staying still'').
The \emph{state complexity} of $M_n$ refers to $|Q_n|$, which is the total number of $M_n$'s inner states, and it is briefly denoted $sc(M_n)$.
A family $\MM$ is said to have \emph{polynomial size} if there exists a polynomial $p$ such that $sc(M_n)\leq p(n)$ holds for all $n\in\nat$.
For more information on the underlying setting of automata families, the reader refers to \cite{Yam22}.

In contrast, one-way models of 2nfa's and 2dfa's, which are called 1nfa's and 1dfa's, take transition function of the form $(Q-Q_{halt})\times\check{\Sigma}\to\PP(Q)$ and $(Q-Q_{halt})\times\check{\Sigma}\to Q$, respectively.
Following \cite{Yam22}, a one-way finite automaton is always assumed to  make no stationary move (or $\varepsilon$-move); that is, its tape head must move only in one direction, to the right, whenever it reads an input symbol (including the  endmarkers). This condition is sometimes called \emph{real time} in the literature.
If $M$ is further allowed to make any $\varepsilon$-move, then we use the specific term ``1.5-way''. See Section \ref{sec:discussion}
for a further discussion.


It is imperative to clarify a few important terminologies associated with ``computation'' of finite automata $M_n$ on input $x$.
We fix an input arbitrarily and consider all surface configurations of each (either one-way or two-way) machine $M_n$ on the input $x$.
A \emph{surface configuration} of $M_n$ on $x$ is of the form $(q,i)$ in $Q\times [0,|x|+1]_{\integer}$ excluding $x$, which indicates that $M_n$ is in inner state $q$ and its tape head is located at cell $i$, assuming that tape cells are indexed by nonnegative integers and that $\rhd$ and $\lhd$ are placed respectively at cell $0$ and cell $|x|+1$.
In the rest of this work, since we deal only with surface configurations, we drop the word ``surface'' altogether from ``surface configurations''.
With the use of configurations, we consider a \emph{computation graph} of $M_n$ on $x$, whose vertices are configurations of $M_n$ on $x$ and a transition from any non-halting configuration to another configuration forms a directed edge. We write $(p,i)\vdash (q,i+1)$ to express an edge of this graph.
An \emph{accepting configuration} is a configuration with an accepting (inner) state. Similarly, a \emph{rejecting configuration} is defined with the use of a rejecting (inner) state.
A \emph{computation path} is a path in a computation graph from the root (i.e., the initial configuration) to a certain leaf (i.e., a halting configuration) if any.
A computation path is \emph{accepting} (resp., \emph{rejecting}) if it ends in a configuration with an accepting (resp., a rejecting) configuration.

We say that $M_n$ \emph{halts} on $x$ if there exists a computation path of finite length in the computation graph of $M_n$ on $x$. In this work, we are interested only in finite automata that always halt on all valid instances. For such a halting 2nfa $M_n$, we say that $M_n$ \emph{accepts} $x$ if $M_n$ starts with $\rhd x\lhd$ and produces at least one accepting computation path, and $M$ is said to \emph{reject} $x$ otherwise.


Let us consider a family $\LL=\{(L_n^{(+)},L_n^{(-)})\}_{n\in\nat}$ of promise problems over a fixed alphabet $\Sigma$.
The set $\Sigma_n$ is a set of \emph{valid} instances in $\Sigma^*$.
Recall that $\Sigma_n$ expresses the set of all valid (or promised) instances over $\Sigma$ for $(L_n^{(+)},L_n^{(-)})$. For any other ``invalid''  instance $x$, when a  machine, say, $M_n$ takes $x$ as an input, we do not require any condition on the behavior of the machine. A family $\MM=\{M_n\}_{n\in\nat}$ of machines over $\Sigma$ is said to \emph{solve} (or \emph{recognize}) $\LL$ if, for any $n\in\nat$, (1) for any $x\in L_n^{(+)}$, $M_n$ accepts $x$ and (2) for any $x\in L_n^{(-)}$, $M_n$ rejects $x$. There may be a case where $M$ does not even halt on invalid instances.

The notation $\twon$ (resp., $\onen$) denotes the collection of all families of promise problems solved by appropriate families of polynomial-size 2nfa's (resp., polynomial-size 1nfa's). Similarly, the notation $\twod$ (resp., $\oned$) is defined using 2dfa's (resp., 1dfa's).


Given an input $xy$, we say that a one-way tape head \emph{crosses} a boundary between $x$ and $y$ in inner state $q$ if $(q_0,0)\vdash^* (p,|\!\rhd\! x|)\vdash (q,|\!\rhd\! x|+1)$ for a certain $p\in Q_n$. A \emph{boundary} is a border between two consecutive tape cells on a tape. The boundary between cells $t-1$ and $t$ is called boundary $t$. A move of a tape head is split into two actions: (i) $M$ reads a tape symbol, (ii) it changes its inner state, and  (iii) it moves its tape head to an adjacent cell. When the tape head move to the next cell, it must cross a boundary with an updated inner state.

\emph{Boundary $0$} is the left border of cell $0$. For any $t\in\nat^{+}$, \emph{boundary $t$} refers to the border between the two cells indexed $t-1$ and $t$. Assume that $M$ returns its tape head back to the start cell and halts. A \emph{crossing sequence at boundary $t$} is a series $(q_{i_1},q_{i_2},\ldots,q_{i_{2k}})$ of inner states such that (i) $M$ first crosses the boundary $t$ from left to right with inner state $q_{i_1}$, (ii) at the $2j$th ($j\in[k]$) turn, $M$ crosses boundary $t$ from right to left with inner state $q_{i_{2j}}$, and (iii) at the $(2j+1)$th turn, $M$ crosses boundary $t$ from left to right with inner state $q_{i_{2j+1}}$. Note that every computation path produces a unique series $(\gamma_0,\gamma_1,\gamma_2,\ldots)$ of crossing sequences $\gamma_t$ at each boundary $t$.


We show a technical but useful lemma for 2nda's by an argument similar to \cite{Kap14}. We assume that all machines in this work make their tape head return to the start cell and halt.

Given a 2nfa $M$ and an input $x$, we define a specific computation path $\gamma_{M}(x)$ as follows. Let $\gamma_{M}(x)$ be the shortest accepting computation path of $M$ on input $x$ if $M$ accepts $x$, and let $\gamma_{M}(x)$ be the shortest rejecting computation path of $M$ on $x$. For convenience, this computation path $\gamma_M(x)$ is called a \emph{critical path} of $M$ on $x$.

In general, when  $\gamma$ is a crossing sequence of the form $(q_{i_1},q_{i_2},\ldots,q_{i_m})$, inner state $q_{i_j}$ ($j\in[m]$)  is said to \emph{appear at location $j$ in $\gamma$}. For an inner state $q$, we say that $q$ \emph{behaves identically at two locations in $\gamma$} if $q$ appears at two different locations, say, $l_1$ and $l_2$ in $\gamma$, and $l_1$ and $l_2$ are both either even or odd. An example of such a case is $\gamma=(q_1,q_3,q_2,q_4,q_1,q_5)$, where $q_1$ behaves identically at locations $1$ and $5$ in $\gamma$.

\begin{lemma}\label{upper-bound-length}
Let $M$ be any 2nfa. Let $x$ be any string. For any boundary $t$, the crossing sequence $\gamma_M(x)$ at boundary $t$ is upper-bounded by $2sc(M)$.
\end{lemma}

\begin{proof}
Let $M=(Q,\Sigma,\{\rhd,\lhd\}, \delta,q_0, Q_{acc}, Q_{rej})$ and let $x$ be any input in $\Sigma^*$. Consider the critical path $\gamma_M(x)$ of $M$ on the input $x$.

Consider a boundary $t$ and a crossing sequence $\tau_t$ at the boundary $t$ on the computation path $\gamma_{M}(x)$. This crossing sequence $\tau_t$ has length at most $2|Q|$ because, otherwise, there exists an inner state $q\in Q$ that behaves identically at two different locations, say, $l_1$ and $l_2$ in $\gamma_{M}(x)$, and thus we can remove all entries in $\tau_t$ between $l_1$ and $l_2$ and form a shorter computation path whose outcome is the same as $\gamma_{M}(x)$. This is a contradiction against the shortest condition of $\gamma_M(x)$.
\end{proof}

It is useful to recall the \emph{branching normal form} from \cite{Yam23b}. A  1nfa $N$ is in the branching normal form if (1) $N$ makes exactly $c$ nondeterministic choices at every step and (2) $N$ produces exactly $c^{|\triangleright x\triangleleft|}$ computation
paths on all inputs $x$ together with two endmarkers, where $c$ is an appropriately chosen constant in $\nat^{+}$.

\begin{lemma}[branching normal form \cite{Yam23b}]\label{branching-normal-form}
Let $M$ be any 1nfa solving a promise problem $(L^{(+)},L^{(-)})$. There exists another 1nfa $N$ such that $N$ is in a branching normal form and $sc(N)\leq 3 sc(M) +1$.
\end{lemma}

\subsection{Unambiguous and Few Computation Paths}\label{sec:various-types}

Unambiguous and few computation paths of nondeterministic machines have played an important role in computational complexity theory. Following early works of \cite{AR88,BJLR91,PTV12,Val76}, we introduce key notions related to ``unambiguity'' and ``fewness'' into the theory of nonuniform polynomial state complexity. Let us recall from Section \ref{sec:FA-CS} that any 1nfa (including 1dfa's) is allowed to have multiple accepting states and multiple rejecting states. Let $\MM=\{M_n\}_{n\in\nat}$ denote any family of (one-way or two-way) nondeterministic finite automata.

\renewcommand{\labelitemi}{$\circ$}
\begin{itemize}\vs{-1}
  \setlength{\topsep}{-2mm}%
  \setlength{\itemsep}{1mm}%
  \setlength{\parskip}{0cm}%

\item Firstly, $\MM$ is called \emph{unambiguous} if, for each index $n\in\nat$ and for any input $x\in\Sigma_n$, there is at most one accepting computation path of $M_n$ on $x$.

\item Similarly, $\MM$ is \emph{weakly-unambiguous} if, for any index $n\in\nat$, for any input $x\in\Sigma_n$, and for any accepting configuration $conf$, there exists at most one computation path from the initial configuration of $M_n$ on $x$ to $conf$.

\item Moreover, $\MM$ is \emph{reach-unambiguous} if, for any index $n\in\nat$, for any input $x\in\Sigma_n$, and for any configuration $conf$, there is at most one computation path from the initial configuration of $M_n$ on $x$ to $conf$.

\item In contrast, $\MM$ is \emph{accept-few}\footnote{This notion is called just ``few'' in \cite{PTV12}. For clarity reason, here we use a slightly different term.} if there exist a polynomial $p$ such that, for any index $n\in\nat$ and for any input $x\in\Sigma_n$, there are at most $p(n,|x|)$ accepting computation paths of $M_n$ on $x$.

\item Lastly, $\MM$ is \emph{reach-few} if there exist a   polynomial $p$ such that, for any index $n\in\nat$, for any input $x\in\Sigma_n$, and for any configuration $conf$ of $M_n$ on $x$, there are at most $p(n,|x|)$ computation paths of $M_n$ on $x$ from the initial configuration to $conf$.
\end{itemize}

The prefix ``reach'' in ``reach-unambiguous'' and ``reach-few'' comes from the ``reachability'' of each configuration from the initial configuration.

It is important to remark that all the above five conditions are applied only to ``valid'' inputs and there is no
requirement for ``invalid'' inputs.

\section{Complexity Classes Defined by One-Way Head Moves}\label{sec:main-result}

We begin with discussing the computational complexity of families of promise problems solved by nonuniform families of one-way nondeterministic finite automata of various types introduced in Section \ref{sec:various-types}.

We stress that, in our model, along any computation path, a tape head of a one-way finite automaton always moves to the right (with no $\varepsilon$-move) until either it reads the right endmarker or it enters a halting inner state before the right endmarker.

In what follows, we will introduce six nonuniform complexity classes of families of promise problems and
prove the class separations among them, as depicted in Figure \ref{fig:class-hierarchy}.

\subsection{Definitions of New Complexity Classes}\label{sec:one-way-head}

It is important to remark that, since tape heads of 1nfa's always move to the right without making any $\varepsilon$-move, we can modify the 1nfa's so that, whenever they fail to enter accepting states until reading the right endmarker $\lhd$, they must enter rejecting states at the time of reading $\lhd$. Moreover, when the 1nfa's enter accepting states even before reading $\lhd$, it is also possible to postpone the timing of acceptance until reading $\lhd$.
The 1nfa's obtained by these modifications halt precisely at reading $\lhd$.

In a way similar to \cite{AR88,BJLR91,PTV12,Val76}, as various subfamilies of $\onen$, we introduce six nonuniform state complexity classes associated with unambiguity and fewness of accepting computation paths given in Section \ref{sec:various-types}.

\renewcommand{\labelitemi}{$\circ$}
\begin{itemize}\vs{-1}
  \setlength{\topsep}{-2mm}%
  \setlength{\itemsep}{1mm}%
  \setlength{\parskip}{0cm}%

\item $\oneu$ consists of all families of promise problems, each  family of which is solved  by an appropriate family of 1-way \emph{unambiguous} finite automata having polynomially many inner states.

\item $\onefew$ consists of all families of promise problems, each  family of which is solved  by an appropriate family of 1-way \emph{accept-few}  finite automata having polynomially many inner states.

\item $\onefewu$ is defined from $\onefew$, whose underlying finite automata are additionally \emph{weakly-unambiguous}.

\item $\onereachu$ is a unique subclass of $\onefew$, whose underlying finite automata are additionally \emph{reach-unambiguous}.

\item $\onereachfew$ is a unique subclass of $\onefew$, whose underlying finite automata are additionally \emph{reach-few}.

\item $\onereachfewu$ is a unique subclass of $\onereachfew$, whose underlying finite automata are additionally \emph{weakly-unambiguous}.
\end{itemize}\vs{-1}

Among the above-mentioned complexity classes, the following inclusion relationships hold.
Figure \ref{fig:class-hierarchy} illustrates these relationships.

\begin{lemma}\label{general-inclusions}
(1) $\oned \subseteq \onereachu  \subseteq \onereachfewu \subseteq \onereachfew \subseteq \onefew$.
(2) $\oned\subseteq \oneu\subseteq \onefewu \subseteq \onefew \subseteq \onen$.
\end{lemma}

\begin{proof}
We prove only two non-trivial inclusions: $\onereachu\subseteq \onereachfewu$ and $\onereachu\subseteq \onefewu$. For any family $\MM$ of 1nfa's, if $\MM$ is reach-unambiguous, then it is also weakly-unambiguous. This implies that $\onereachu\subseteq \onefewu$. Moreover, if $\MM$ is reach-unambiguous, then it is also reach-few. Thus, $\onereachu\subseteq \onereachfewu$ follows instantly.
\end{proof}

Given a class $\CC$ of families of promise problems, we write $\co\CC$ for the complexity class $\{\LL\mid \co\LL\in \CC\}$.
The following equality is easily proven for $\onereachfew$.

\begin{lemma}\label{reachfew-complement}
$\onereachfew = \co\onereachfew$.
\end{lemma}

\begin{proof}
Let $\LL=\{(L_n^{(+)},L_n^{(-)})\}_{n\in\nat}$ be any family in $\onereachfew$. Consider a family $\MM=\{M_n\}_{n\in\nat}$ of polynomial-size 1nfa's solving $\LL$ such that each $M_n$ is reach-few and accept-few. Since $M_n$ is one-way, it can be assumed to enter either accepting or rejecting states at reading the right endmarker. We thus define another 1nfa $\overline{M}_n$ to be the one obtained from $M_n$ by swapping $Q_{acc,n}$ and $Q_{rej,n}$ of $M_n$. Clearly, $\overline{M}_n$ solves $(L_n^{(-)},L_n^{(+)})$, which belongs to $\co\LL$. Note that, at each step of a computation, the number of distinct configurations of a 1nfa $M_n$ is at most a certain polynomial in $n$. Thus, $\overline{M}_n$ is also reach-few and accept-few. Therefore, $\co\LL$ belongs to $\co\onereachfew$.
\end{proof}

\subsection{Class Separations Concerning 1ReachU}\label{sec:separate-onereachu}

It is known that $\onen$ is not closed under complementation \cite{SS78}.  This statement is logically equivalent to $\co\onen\nsubseteq \onen$ and this statement was extended in \cite{Yam23b} to $\co\oneu\nsubseteq \onen$. From this separation follows the non-closure property of $\oneu$ under complementation \cite{Yam23b}.
Additionally, we wish to demonstrate that $\co\onereachu\nsubseteq \onen$ by utilizing the notion of Kolmogorov complexity.

\begin{theorem}\label{co-oneu-vs-onen}
$\co\onereachu\nsubseteq \onen$.
\end{theorem}


To prove the above theorem, let us introduce several important notations. Given numbers $k,m,n\in\nat^{+}$ with $m\geq n$ and $i_1,i_2,\ldots,i_k\in[n]$, we use the notation
$\dbraleft i_1,i_2,\ldots,i_k \dbraright_m$ to be $1^{i_1}0^{m-i_1}0 1^{i_2}0^{m-i_2}0 \cdots 0 1^{i_k}0^{m-i_k}$, where we treat $0^0$ as $\lambda$.
Here, the value $k$ is called the \emph{size} of $\dbraleft i_1,i_2,\ldots,i_k \dbraright_m$.
Let $r$ denote any string of the form $\dbraleft i_1,i_2,\ldots,i_k \dbraright_m$. We then obtain $|r|= |\dbraleft i_1,i_2,\ldots,i_k \dbraright_m| = km+k-1$.
Given any index $e\in[k]$, the notation $(r)_{(e)}$ denotes $i_e$ and $Set(r)$ stands for the set $\{i_1,i_2,\ldots,i_k\}$.
When we need to distinguish multiple occurrences of the same number in $r$,  we use another notation $MSet(r)$ to express its \emph{multi set}.
We remark that $|Set(r)|\leq |MSet(r)|=k$.
Let $B_n(m,k)$ denote the set of all strings of the form $\dbraleft i_1,i_2,\ldots,i_k \dbraright_m$ with  $i_1,i_2,\ldots,i_k\in[n]$.


\vs{-2}
\begin{proofof}{Theorem \ref{co-oneu-vs-onen}}
We introduce the family $\LL_1=\{(L_n^{(+)},L_n^{(-)})\}_{n\in\nat}$ of promise decision problems by setting $L_n^{(+)} = \{u\# v\mid u,v\in B_n(n,n), \exists e\in[n]( (u)_{(e)}\neq (v)_{(e)} )\}$ and $L_n^{(-)}=\{u\# v\mid u,v\in B_n(n,n), \forall e\in[n]( (u)_{(e)} = (v)_{(e)} )\}$ for any index $n\in\nat$.
Notice that $|u|=|\dbraleft i_1,i_2,\ldots,i_n \dbraright_n| = n(n+1)-1$ and $|u\# u| \leq 2n(n+1)-1$.
The complement family of $\LL_{1}$,  $\co\LL_{1}$, is $\{(L_n^{(-)},L_n^{(+)})\}_{n\in\nat}$. In what follows, we wish to prove that (1) $\LL_1\in\onereachu$ and (2) $\co\LL_1\notin \onen$. Since $\co\LL_{1}\in\co\onereachu$, ``$\co\LL_1\notin\onen$''
derives the theorem.

(1) Let us first prove that $\LL_1$ belongs to $\onefewu$. This can be shown by constructing a family of polynomial-size, accept-few, weakly-unambiguous  1nfa's $N_n$ that solves $\LL_1$.
We make each $N_n$ have $n$ accepting inner states, say,  $Q_{acc,n} = \{\hat{q}_1,\hat{q}_2,\ldots,\hat{q}_n\}$.
On input $x$ of the form $r_1\# r_2$ with $r_1,r_2\in B_n(n,n)$,  guess (i.e., nondeterministically choose) an index $e\in[n]$, read $r_1$, remember the $e$th entry $i_e = (r_1)_{(e)}$ by entering inner states of the form $(i_e,e)$, read $r_2$, and check whether  $i_e$ coincides with $(r_2)_{(e)}$. If so, enter the $e$th accepting state $\hat{q}_e$.
It is easy to show that $N_n$ solves $(L_n^{(+)},L_n^{(-)})$ correctly. Notice that, for each accepting inner state $\hat{q}_e$, there are at most one accepting computation path leading to $\hat{q}_e$. Thus, $N_n$ is accept-few and also weakly-unambiguous. Since $\{N_n\}_{n\in\nat}$ solves $\LL_1$, it belongs to $\onereachu$.

(2) To verify that $\co\LL_{1}\notin \onen$, let us assume that $\co\LL_1\in\onen$ and take a family $\MM=\{M_n\}_{n\in\nat}$ of polynomial-size 1nfa's that solves $\co\LL_{1}$.
For each index $n\in\nat$, we write $Q_n$ for the set of all inner states of $M_n$. The cardinality of such the set $Q_n$ is  upper-bounded by an appropriately fixed polynomial, say, $p$; namely, $|Q_n|\leq p(n)$ for all $n\in\nat$.
For our convenience, we express all elements of $Q_n$ as numbers in  $[0,|Q_n|-1]_{\integer}$.
Take a sufficiently large $n$ satisfying $p(n)<2^n$.
Given any $n\in\nat$ and any string $u\in B_n(n,n)$, we define $S_n(u)$ to be the set of all inner states $q$ in $Q_n$ such that, in a certain accepting computation path of $M_n$ on input $u\# u$, $M_n$ enters $q$ just after reading $u\#$.

Recall from Section \ref{sec:number-promise} the notion of the (unconditional) Kolmogorov complexity $C(w)$ of a string $w$. There exists a string $u$ in $B_n(n,n)$ such that $C(u)\geq n\log{n}-1$. If there is no such $u$, we then define a mapping from $B_n(n,n)$ to $\Sigma^*$ as follows. Given each element $u$ in $B_n(n,n)$, we assign to $u$ a ``minimal'' program $p_u$ that the universal Turing machine $U$ takes $p_u$ as an input and produces $u$ on its output tape. Since $C(u)=|p_r|$, we obtain $|p_r|<n\log{n}-1$.
This implies $\|B_n(n,n)\|\leq \|\{0,1\}^{<n\log{n}-1}\| = 2^{n\log{n}-1}$. This is a contradiction.

Now, we evaluate the value $C(u)$. Note that the length of the binary expression of $q$ is $O(\log{n})$ since $|Q_n|\leq p(n)$. Take an inner state $q\in S_n(u)$ making $C(q)$ be the smallest.
Let us design a deterministic algorithm, say, $\BB$ that begins with the empty input $\varepsilon$ and produces $u$ on an output tapes.
The algorithm $\BB$ recursively picks up strings $w\in B_n(n,n)$ one by one,  runs $M_n$ on $w\# w$ starting with the inner state $q$.
If $M_n$ enters no accepting states on all computation paths, then $\BB$ outputs $w$ and halts. Otherwise, $\BB$ picks up another $w$ and continues the above procedure.
Let $r_0$ denote the binary encoding of this algorithm $\BB$. Notice $C(q)\leq |bin(q)|$.
Since $\BB$ requires the knowledge of $(n,q)$ to run, it follows that $C(u)\leq |bin(n)|+|bin(q)|+ |r_0| + O(1) \leq O(\log{n})$, which is in contradiction to $C(u)\geq n\log{n}-1$.
\end{proofof}


An instant corollary of Theorem \ref{co-oneu-vs-onen} is the non-closure property under complementation for two complexity classes: $\onereachu$, $\onereachfewu$, $\onefewu$, and $\onefew$.

\begin{corollary}\label{complement-various}
\renewcommand{\labelitemi}{$\circ$}
$\onereachu\neq \co\onereachu$, $\onereachfewu\neq \co\onereachfewu$, $\onefewu\neq \co\onefewu$, and $\onefew\neq \co\onefew$.
\end{corollary}

\begin{proof}
Assume that $\onereachu=\co\onereachu$. Since $\onereachu\subseteq \onen$, we obtain $\co\onereachu\subseteq \onen$. This contradicts Theorem \ref{co-oneu-vs-onen}. Thus, we conclude that $\onereachu\neq\co\onereachu$. The same argument works for $\onereachfewu$, $\onefewu$, and $\onefew$ since these classes include $\onereachu$.
\end{proof}


Another implication of Theorem \ref{co-oneu-vs-onen} is the class separations between $\onereachfew$ and $\onefew$ as well as between $\oned$ and $\onereachu$.

\begin{corollary}\label{onereachfew-vs-oneu}
(1) $\oned\neq \onereachu$.  (2) $\onereachfewu \neq \onereachfew$. (3) $\onereachfew\neq\onefew$.
\end{corollary}

\begin{proof}
(1) Obviously, $\onereachu$ contains $\oned$.  If $\oned=\onereachu$, then we obtain $\oned=\co\onereachu$ because of $\oned=\co\oned$. This implies that $\co\onereachu\subseteq \onen$, a contradiction with Theorem \ref{co-oneu-vs-onen}.

(2) Since $\onereachfew=\co\onereachfew$ by Lemma \ref{reachfew-complement}, Theorem \ref{co-oneu-vs-onen} then implies that $\onereachfewu\neq \onereachfew$.

(3) If $\onereachfew=\onefew$, then Corollary \ref{complement-various}(1) implies that $\onereachfew=\onefew\neq \co\onefew=\co\onereachfew$. However, this is in  contradiction to Lemma \ref{reachfew-complement}.
\end{proof}

Corollary \ref{onereachfew-vs-oneu}(2) immediately draws a conclusion that $\onereachfew\neq \onereachu$. This is not true for the corresponding log-space complexity classes: $\mathrm{ReachFewL}$ and $\mathrm{ReachUL}$ \cite{GSTV11}.

\subsection{Class Separations Concerning 1U}\label{sec:separation-1U}

We next look into the computational complexity of $\oneu$, which is a one-way finite automata analogue of the complexity class $\up$. From Lemma \ref{general-inclusions} follows $\onereachfewu\subseteq \onefewu$.
In what follows, we assert that $\onefewu$ is an optimal upper bound of $\onereachfewu$ by demonstrating the difference between $\oneu$ and $\onefewu$.

\begin{theorem}\label{u-fewu}
$\onereachu\nsubseteq \oneu$.
\end{theorem}

It was noted in \cite{Yam23b} that there is a close connection
between one-tape linear-time advised Turing machines and nonuniform finite automata families, in particular, a connection
between $1\mbox{-}\mathrm{C}_{=}\mathrm{LIN}/lin$ and $\onecequal$. This connection was adapted to prove the separation of $\onecequal$ from $\onen$. A key to this separation is a lemma (Lemma 12) of \cite{Yam23b} expressing a specific feature of families of promise problems in $\onecequal$.
To prove Theorem \ref{u-fewu}, we wish to generalize this key lemma to Lemma \ref{property-1pfa} and apply it in the proof of the theorem.

For the generalization of the key lemma in \cite{Yam23b}, we first need to convert the current model of 1nfa into the following model of \emph{one-way probabilistic finite automata} (or 1pfa's, for short).
A 1pfa $N$ has the form $(Q,\Sigma,\{\triangleright,\triangleleft\}, \nu_{ini}, \{M_{\sigma}\}_{\sigma\in\check{\Sigma}}, Q_{acc},Q_{rej})$ with a $|Q|$-dimensional (row) vector $\nu_{ini}$ and a set of $|Q|\times |Q|$ stochastic\footnote{A square matrix of nonnegative real entries is called \emph{stochastic} if every row of the matrix sums up to exactly $1$.}
matrices $M_{\sigma}$ for each symbol $\sigma$. Given a string $x=x_1x_2\cdots x_n$ in $\check{\Sigma}^*$, we abbreviate as $M_x$ the matrix product $M_{x_1}M_{x_2}\cdots M_{x_n}$. We fix an accepting inner state $q_{acc}\in Q_{acc}$.
We further define a $|Q|$-dimensional (row) vector $\xi_{acc}(q_{acc})=(\xi_q)_{q\in Q}$ associated with $q_{acc}$ by setting $\xi_{q_{acc}}=1$ and $\xi_q=0$ for all $q\in Q-\{q_{acc}\}$.
The acceptance probability of $N$ on input $x$ leading to $q_{acc}$ is denoted  $p_{acc}(x:q_{acc})$, which is defined as $\nu_{ini}M_{\triangleright x\triangleleft}\xi_{acc}(q_{acc})^T$, where $T$ refers to the ``transpose'' of a vector.
Moreover, we set $p_{acc}(x)$ to be $\sum_{q\in Q_{acc}} p_{acc}(x:q)$, which is the acceptance probability of $N$ on input $x$.

Given a 1nfa $M = (Q,\Sigma,\delta,q_0, Q_{acc},Q_{rej})$ in the branching normal form, consider the following conversion of it into the corresponding 1pfa $N$. Take a constant $c\geq1$ fir which $M$ makes exactly $c$ nondeterministic choices at every step. We further expand $\delta$ to the domain $Q\times \check{\Sigma}$ (from $(Q-Q_{halt})\times \check{\Sigma}$) by assigning $\delta(q,\sigma) = \{q\}$ for all $q\in Q_{halt}$ and all $\sigma\in\check{\Sigma}$.
We turn a nondeterministic choice of $M$ to a probabilistic choice of $N$ made with probability $1/c$ by defining $M_{\sigma}=(a_{ij})_{i,j\in Q}$ as $a_{ij}=1/c$ if $j\in \delta(i,\sigma)$ and $a_{ij}=0$ otherwise.
Since $|\delta(i,\sigma)|=c$ for any $(i,\sigma)$, the $i$th row of $M_{\sigma}$ has exactly $c$ nonzero entries, and thus all entries in the $i$th row sum up to $1$. This guarantees the stochasticity of $M_{\sigma}$.
Moreover, the total number of accepting computation paths of $M$ on $x$ matches the total acceptance probability of $N$ on $x$.

\begin{lemma}\label{property-1pfa}
Let $\MM=\{M_n\}_{n\in\nat}$ be a family of polynomial-size 1pfa's over alphabet $\Sigma$ in which each $M_n$ makes exactly $c$ probabilistic  choices with equal probability at every step for a fixed constant $c\in\nat^{+}$ (independent of $n$). For each index $n\in\nat$, let $D_n$ denote a set of valid instances in $\Sigma^*$.
Let $f:\nat^3\times Q_{acc,n}\to[0,1]$ be a function, where $Q_{acc,n}$ is  the set of all accepting inner states of $M_n$.
In this case, there exists a polynomial $p$ that satisfies the following statement.
(*) For any numbers $n,m,l\in\nat^{+}$ with $l\leq m-1$ and any accepting inner state $q\in Q_{acc,n}$, there exists a subset $S$ of $A_{n,m,l,q}= \{x\in\Sigma^{m-l}\mid \exists z\in\Sigma^l\:[ xz\in D_n \wedge  p_{acc,n}(xz:q)=f(n,|xz|,|z|,q)] \}$ with $|S|\leq p(n)$ such that, for any $y\in\Sigma^l$, if $wy\in D_n$ and $p_{acc,n}(wy:c)=f(n,m,l,q)$ for all $w\in S$, then  $p_{acc,n}(xy:q)=f(n,m,l,q)$ holds for all $x\in A_{n,m,l,q}$ with $xy\in D_n$.
\end{lemma}

\begin{proof}
Let $\MM=\{M_n\}_{n\in\nat}$ be any given family of polynomial-size 1pfa's over alphabet $\Sigma$, where each $M_n$ has the form $(Q_n,\Sigma,\{\triangleright,\triangleleft\}, \nu_{ini}^{(n)}, \{M^{(n)}_{\sigma}\}_{\sigma\in\check{\Sigma}}, Q_{acc,n},Q_{rej,n})$, making exactly $c$ probabilistic choices with equal probability at every step.
Notice that, for any $m,n\in\nat$, any input $x\in\Sigma^m\cap D_n$, and any accepting inner state $q$ of $M_n$, $p_{acc,n}(x:q)$ equals $\nu_{ini}^{(x)}M^{(n)}_{\triangleright x\triangleleft} \xi_{acc,n}(q)^T$.
We remark that, when $M_n$ reads the right endmarker $\triangleleft$ on any valid input $x$, each computation path of $M_n$ on $x$ is generated with probability exactly $c^{|x|+2}$. Since $\MM$ has polynomial size,
there is a polynomial $p$ satisfying $|Q_n|\leq p(n)$ for all $n\in\nat$.

We fix $n$ arbitrarily and take $m,l\in\nat^{+}$ with $l\leq m-1$ and $q\in Q_{acc,n}$. The notation $A_{n,m,l,q}$ denotes the set $\{x\in\Sigma^{m-l}\mid \exists z\in\Sigma^l\:[ xz\in D_n \wedge  p_{acc,n}(xz:q) = f(n,|xz|,|z|,q)] \}$.
We consider the set $V_{n,m,l,q} = \{\nu_{ini}^{(n)}M^{(n)}_{\triangleright x} \mid x\in A_{n,m,l,q}\}$ of $|Q_n|$-dimensional vectors.
Note that $|V_{n,m,l,q}|\leq |\Sigma|^{m-l}$.
We choose a maximal subset of linearly-independent vectors in $V_{n,m,l,q}$. Since all vectors in $V_{n,m,l,q}$ are $|Q_n|$-dimensional, such a subset is expressed as  $\{\nu_{ini}^{(n)}M_{\triangleright w} \mid w\in S\}$ for a certain subset $S$ of $\Sigma^{m-l}$ of size at most $|Q_n|$. It then follows that, for any $x$ with $|x|=m-l$, by the choice of $\{\nu_{ini}^{(n)}M_{\triangleright w} \mid w\in S\}$, there exists a series $\{\alpha^{(x)}_{w}\}_{w\in S}$ of real numbers satisfying that $\nu_{ini}^{(n)}M^{(n)}_{\triangleright x} = \sum_{w\in S} \alpha^{(n)}_{w} \nu_{ini}^{(n)}M^{(n)}_{\triangleright w}$.
Here, we claim that $\sum_{w\in S}\alpha^{(x)}_{w}=1$. To prove this claim, let $\nu_{ini}^{(n)}M^{(n)}_{\triangleright w} = (a_{wi})_{i\in[m]}$ and $\nu^{(n)}_{ini}M^{(n)}_{\triangleright x} = (b_{xi})_{i\in[m]}$. It then follows that $(b_{xi})_{i\in[m]} = (\sum_{w\in S}\alpha^{(x)}_{w}a_{wi})_{i\in[m]}$. From this follows $\sum_{i\in[m]}b_{xi} = \sum_{w\in S} (\sum_{i\in[m]} \alpha^{(x)}_{w}a_{wi}) = \sum_{w\in S}\alpha^{(x)}_{w} (\sum_{i\in[m]}a_{wi})$.
The stochasticity of $M^{(n)}_{\triangleright x}$ and $M^{(n)}_{\triangleright w}$ implies that $\sum_{i\in[m]}b_{xi} = \sum_{i\in[m]}a_{wi}=1$. Therefore, we conclude that $\sum_{w\in S}\alpha^{(x)}_{w}=1$.

We fix $y\in\Sigma^l$ and $q\in Q_{acc,n}$ arbitrarily and assume that $wy\in D_n$ and $p_{acc,n}(wy:q)=f(n,m,l,q)$ for all $w\in S$.
Take any string $x\in A_{n,m,l,q}$ satisfying $xy\in D_n$. It then follows that  $p_{acc,n}(xy:q) = \nu_{ini}^{(n)}M^{(n)}_{\triangleright xy \triangleleft}\xi_{acc,n}(q)^T = (\nu_{ini}^{(n)}M_{\triangleright x}) (M_{y\triangleleft}\xi_{acc,n}(q)^T)
= (\sum_{w\in S} \alpha^{(x)}_{w} \nu_{ini}^{(n)}M^{(n)}_{\triangleright w}) (M^{(n)}_{y\triangleleft}\xi_{acc,n}(q)^T)
= \sum_{w\in S}\alpha^{(x)}_{w} ( \nu_{ini}^{(n)} M^{(n)}_{\triangleright wy \triangleleft}\xi_{acc,n}(q)^T)
= \sum_{w\in S}\alpha^{(x)}_{w} p_{acc,n}(wy:q)$.
The last term is equal to $\sum_{w\in S}\alpha^{(x)}_{w} f(n,m,l,q) = (\sum_{w\in S}\alpha^{(x)}_{w}) f(n,m,l,q)$.
Since $\sum_{w\in S}\alpha^{(x)}_{w}=1$, we obtain $p_{acc,n}(xy:q) = f(n,m,l,q)$, as requested.
\end{proof}

In Lemma \ref{property-1pfa}, the use of the equality ``$=$'' in the above formula $p_{acc,n}(xz:q)=f(n,|xz|,|z|,q)$ is crucial because, if ``$=$'' is replaced by, e.g., ``$\leq$'', then we may not be able to prove a statement similar to (*) in the lemma.


Now, we provide the proof of Theorem \ref{u-fewu} using Lemma \ref{property-1pfa}.

\begin{proofof}{Theorem \ref{u-fewu}}
As an example family that falls in $\onereachu$, let us  recall the family $\LL_1 = \{(L_n^{(+)},L_n^{(-)})\}_{n\in\nat}$ introduced in the proof of Theorem \ref{co-oneu-vs-onen}.
We thus need to verify that $\LL_1\notin \oneu$.
Toward a contradiction, we assume the existence of a  family $\MM=\{M_n\}_{n\in\nat}$ of polynomial-size unambiguous 1nfa's that solves $\LL_1$.
By Lemma \ref{branching-normal-form}, we can assume that $M_n$ makes exactly $c$ nondeterministic choices at every step. Moreover, we can keep $M_n$ unambiguous on all valid instances.
Since $M_n$ is unambiguous, it is possible to assume that $|Q_{acc,n}|=1$ for all $n\in\nat$. Let $Q_{acc,n} = \{q_{acc}\}$.
For clarity reason, we denote $L_n^{(+)}\cup L_n^{(-)}$ by $D_n$.
By the definition of $(L_n^{(+)},L_n^{(-)})$, it follows that $x\in D_n$ implies $|x|=2n(n+1)-1$ since $x$ is of the form $r_1\# r_2$ for certain strings $r_1,r_2\in B_n(n,n)$.
Since $\MM$ has polynomial size, there is a certain polynomial $p'$ satisfying $|Q_n|\leq p'(n)$ for all $n\in\nat$.

As noted earlier, we can convert $M_n$ into a 1pfa of the form $(Q_n,\Sigma,\{\triangleright,\triangleleft\}, \nu_{ini}^{(n)}, \{M^{(n)}_{\sigma}\}_{\sigma\in\check{\Sigma}}, Q_{acc,n},Q_{rej,n})$ with the following acceptance criteria: for any $x$, $x\in L_n^{(+)}$ implies $p_{acc,n}(x)=1/c^{|x|+2}$ and $x\in L_n^{(-)}$ implies $p_{acc,n}(x)=0$. It is important to remark that this 1pfa $M_n$ has only one accepting inner state $q_{acc}$, and thus  $p_{acc,n}(x:q_{acc})$ matches $p_{acc,n}(x)$.

For any $n,m,l\in\nat$, we further define $f(n,m,l,q_{acc})$ to be $1/c^{m+2}$.
By Lemma \ref{property-1pfa}, there is a polynomial $p$ that satisfies the lemma. In what follows, we fix a number $n\in\nat$ that guarantees that  $p(n)+p'(n)+1 < n^n$. Let $m=2n(n+1)-1$ and $l=n(n+1)-1$.
Moreover, we fix $r_0= \dbraleft 1,1,\ldots,1 \dbraright_n$ in $B_n(n,n)$. Note that $|r_0|=l$ and $|r\#r_0|=m$ for all elements $r\in B_n(n,n)$.
The notation $A_{n,m,l,q_{acc}}$ denotes the set $\{r\# \mid r\in B_n(n,n),  \exists r'\in B_n(n,n)\:[ p_{acc,n}(r\# r')= f(n,|r\# r'|,|r'|,q_{acc})] \}$.
The lemma ensures the existence of a set $S\subseteq A_{n,m,l,q_{acc}}$ with $|S|\leq p(n)$ that satisfies the statement (*) of the lemma.  We choose a string $y\in B_n(n,n)$ satisfying $y\#\notin S\cup\{r_0\#\}$.
This is possible because $|S|\leq p(n)$ and $|B_n(n,n)|=n^n$.
Notice that $y\#$ is in $A_{n,m,l,q_{acc}}$
since $p_{acc,n}(y\# r_0)=1/c^{|y\# r_0|+2}$.
For any $r\#\in S$, since $r\neq y$, we obtain $r\#y\in L_n^{(+)}$, and thus $p_{acc,n}(r\# y) = f(n,m,l,q_{acc})$ follows.
The statement (*) of the lemma then concludes that $p_{acc,n}(r\# y) = f(n,m,l,q_{acc})$ for all $r\#\in A_{n,m,l,q_{acc}}$.
As a special case, we choose $y$ as $r$, and thus $p_{acc,n}(y\# y)=f(n,m,l,q_{acc})$ follows.
Therefore, $y\# y$ belongs to $L_n^{(+)}$, a contradiction with the fact that $y\# y\in L_n^{(-)}$.

This completes the proof.
\end{proofof}


Pavan \etalc~\cite{PTV12} demonstrated that $\mathrm{ReachFewL}$ is included in $\ul\cap\co\ul$. A simple analogy could suggest that $\onereachfew$ might be included in $\oneu\cap\co\oneu$. However, this analogy does not hold because we instantly obtain $\onereachfew\nsubseteq \oneu$ from Theorem \ref{u-fewu}.

\subsection{Class Separations Concerning 1FewU}

Let us recall the key lemma, Lemma \ref{property-1pfa}, of Section \ref{sec:separation-1U} and we intend to apply it to show another separation between $\onefewu$ and $\onen$.

\begin{theorem}\label{onefewu-vs-onen}
$\onen \nsubseteq \onefewu$.
\end{theorem}

Given $n\in\nat^{+}$, let us consider an $n\times n$ matrix $D=(a_{ij})_{i,j\in[n]}$ whose $(i,j)$-entry $a_{ij}$ is in $[n]$. For each index $i\in[n]$, the notation $D(i)$ denotes the $i$th row $(a_{i1},a_{i2},\ldots,a_{in})$ of $D$. Associated with this row $D(i)$, we encode it into a unique string $\dbraleft D(i)\dbraright_n$ of the form $\dbraleft a_{i1}, a_{i2},  \ldots , a_{in} \dbraright_n$. Moreover, the encoding of $D$ is given as $\dbraleft D \dbraright_n = \dbraleft D(1) \dbraright_n \# \dbraleft D(2) \dbraright_n \# \cdots \# \dbraleft D(n) \dbraright_n$. Note that $|\dbraleft D(i) \dbraright_n| = n(n+1)-1$ for any $i\in[n]$ and $|\dbraleft D \dbraright_n| = n^2(n+1)-1$.
For convenience, $T_n$ denotes the set $\{D\mid D=(a_{ij})_{i,j\in[n]} \text{ is an $n\times n$ matrix with $a_{ij}\in[n]$ for any $i,j\in[n]$ }\}$.

\begin{proofof}{Theorem \ref{onefewu-vs-onen}}
Let $n\in\nat^{+}$, take a matrix $D=(a_{ij})_{i,j\in[n]}$ in $T_n$, and consider its encoding $\dbraleft D \dbraright_n$. Using a choice function $\sigma:[n]\to[n]$, we define $sum_D(\sigma)=\sum_{i=1}^{n}a_{i\sigma(i)}$.

For each index $n\in\nat$, we define $L_n^{(+)}= \{ \dbraleft D \dbraright_n \#^2 \dbraleft D' \dbraright_n \mid D,D'\in T_n, \exists \sigma,\sigma'\:[sum_D(\sigma) = sum_{D'}(\sigma')] \}$ and $L_n^{(-)}= \{ \dbraleft D \dbraright_n \#^2 \dbraleft D' \dbraright_n \mid D,D'\in T_n, \forall \sigma,\sigma'\:[sum_D(\sigma) \neq sum_{D'}(\sigma')] \}$. We then set $\LL_2$ to be the family $\{(L_n^{(+)},L_n^{(-)})\}_{n\in\nat}$.
Hereafter, we intend to prove (1) $\LL_2\in\onen$ and (2) $\LL_2\notin\onefewu$ to complete the proof.

(1) We first show that $\LL_2\in\onen$. We arbitrarily fix $n\in\nat$ and consider the following 1nfa $N_n$.
Let an input have the form $\dbraleft D\dbraright_n \#^2 \dbraleft D'\dbraright_n$ with $D,D'\in T_n$, where $D=(a_{ij})_{i,j\in[n]}$ and $D'=(a'_{ij})_{i,j\in[n]}$. While reading $\dbraleft D\dbraright_n$, for each $i\in[n]$, nondeterministically choose one entry, say, $a_{ij_i}$ ($j_i\in[n]$) from the $i$th row $D(i)$ and compute $\bar{a}= \sum_{i\in[m]}a_{ij_i}$. Remember $\bar{a}$ and cross $\#^2$ to $\dbraleft D'\dbraright_n$. Repeat the same procedure to compute $\bar{a}' = \sum_{i\in[n]}a'_{ik_i}$. Check whether $\bar{a}=\bar{a}'$. If so, enter an accepting inner state; otherwise, enter a rejecting inner state.
This machine $N_n$ is obviously a 1nfa and correctly solves $(L_n^{(+)},L_n^{(-)})$. Therefore, $\LL_2$ belongs to $\onen$.

(2) Next, we want to show that $\LL_2\notin \onefewu$ by way of contradiction.  Assume that $\LL_2\in\onefewu$ and take a family $\MM=\{M_n\}_{n\in\nat}$ of polynomial-size, accept-few,
weakly-unambiguous 1nfa's that correctly solves $\LL_2$.
It is possible to assume by Lemma \ref{branching-normal-form} that $M_n$ makes exactly $c$ nondeterministic choices at every step
for a certain fixed constant $c\geq1$.
As done in Section \ref{sec:separation-1U}, we can convert $M_n$ into a 1pfa with the following acceptance criteria: for any $x\in L_n^{(+)}\cup L_n^{(-)}$, $x\in L_n^{(+)}$ implies that there is an inner state $q\in Q_{acc,n}$ for which $p_{acc,n}(x:q)=1/c^{|x|+2}$ and $x\in L_n^{(-)}$ implies that, for all $q\in Q_{acc,n}$, $p_{acc,n}(x:q)=0$.
We then define $f(n,m,l,q)=1/c^{m+2}$.
Given any $D'\in T_n$, $C_n(D')$ denotes the set $\{D\in T_n\mid \forall \sigma,\sigma'\:[ sum_{D}(\sigma)\neq sum_{D'}(\sigma')]\}$.

To lead to a desired contradiction, we apply Lemma \ref{property-1pfa} and take a polynomial $p$ that satisfies the lemma.
Let $D_0$ denote the $n\times n$ matrix whose entries are all $1$. Clearly, $D_0$ is in $T_n$.
Let $m=n^2(n+1)-1$ and $l=n(n+1)-1$.
Given any $q\in Q_{acc,n}$, we introduce the set $A_{n,m,l,q} = \{ \dbraleft D \dbraright_n\#^2 \mid D\in T_n, \exists D'\in T_n\: [ p_{acc,n}(\dbraleft D \dbraright_n \#^2 \dbraleft D' \dbraright_n :q) = f(n,m,l,q)] \}$.
For convenience, we introduce the notation $A_{n,m,l,q}^{(-)}$ defined to be $\{D\in T_n\mid \dbraleft D \dbraright_n\#^2\in A_{n,m,l,q}\}$.
We then take a subset $S$ of $A_{n,m,l,q}$ of size at most $p(n)$ that satisfies the statement (*) of the lemma.
We choose an arbitrary matrix $D\in A^{(-)}_{n,m,l,q}$ outside of   $\{D''\in C_n(D') \mid D'=D_0 \text{ or } \dbraleft D' \dbraright_n\#^2 \in S\}$.
This $D$ satisfies the following property:
for any $D'\in T_n$, if $\dbraleft D' \dbraright_n\#^2$ is in $S$, then  $p_{acc,n}(\dbraleft D' \dbraright_n\#^2 \dbraleft D \dbraright_n:q) = f(n,m,l,q)$ since $D\notin C_n(D')$.
Therefore, by the conclusion of the statement (*), $p_{acc,n}(\dbraleft D' \dbraright_n\#^2 \dbraleft D \dbraright_n:q) = f(n,m,l,q)$ follows for all  $D' \in A^{(-)}_{n,m,l,q}$. In particular, we set $D' = D$ and obtain  $p_{acc,n}(\dbraleft D \dbraright_n\#^2 \dbraleft D \dbraright_n:q) = f(n,m,l,q)$.
This implies that $\dbraleft D \dbraright_n\#^2 \dbraleft D \dbraright_n\in L_n^{(+)}$. Obviously, this is in contradiction to the definition of $L_n^{(+)}$.
\end{proofof}

We can conclude the following, rather weak, statement from Theorem \ref{onefewu-vs-onen}.

\begin{corollary}
Either $\onefew\neq\onen$ or $\onefewu\neq\onefew$ (or both) is true.
\end{corollary}


\section{Complexity Classes Defined by Two-Way Head Moves}\label{sec:two-way-head}

In Section \ref{sec:one-way-head}, we have introduced six nonuniform state complexity classes situated in between $\oned$ and $\onen$. By replacing underlying one-way finite automata defining those complexity classes with two-way finite automata, we naturally obtain  the corresponding six nonuniform state complexity classes: $\twou$, $\tworeachu$, $\tworeachfewu$, $\tworeachfew$, $\twofewu$, and $\twofew$.
Similar to Lemma \ref{general-inclusions} for the case of one-way head moves, these classes satisfy the following clear inclusion relationships:
(1) $\twod \subseteq \tworeachu  \subseteq \tworeachfewu \subseteq \tworeachfew \subseteq \twofew$ and
(2) $\twod\subseteq \twou\subseteq \twofewu \subseteq \twofew \subseteq \twon$.

\subsection{Cases of Logarithmic and Polynomial Ceilings}\label{sec:bounded-ceilings}

A family $\LL=\{(L_n^{(+)},L_n^{(-)})\}_{n\in\nat}$ of promise problems is said to \emph{have a polynomial ceiling} if there exists a polynomial $p$ for which $\Sigma_n \subseteq \Sigma^{\leq p(n)}$ holds for all $n\in\nat$, where $\Sigma_n$ denotes the set $L_n^{(+)}\cup L_n^{(-)}$. Similarly, we can define the notion of \emph{exponential ceiling} (resp., \emph{logarithmic ceiling}) simply by replacing the term ``polynomial'' with ``exponential'' (resp., ``logarithmic function''). More generally, given a function $g:\nat\to\nat$, we say that $\LL$ has a \emph{$g(n)$-ceiling} if $\Sigma_n\subseteq \Sigma^{\leq g(n)}$ holds for any number $n\in\nat$.

For two functions $f$ and $g$ on $\nat$ (i.e., from $\nat$ to $\nat$), we say that $g$ \emph{majorizes} $f$ (denoted $g\geq f$) if $g(n)\geq f(n)$ holds for all $n\in\nat$. Consider a family $\LL=\{(L_n^{(+)},L_n^{(-)})\}_{n\in\nat}$ of promise problems having an $f(n)$-ceiling.  It then follows that $L_n^{(+)} =\{x\in L_n^{(+)}\mid |x|\leq f(n)\}$ and $L_n^{(-)}=\{x\in L_n^{(-)}\mid |x|\leq f(n)\}$.
For any function $g$ on $\nat$, if $g$ majorizes $f$, then we obtain  $\Sigma_n\subseteq \Sigma^{\leq f(n)} \subseteq \Sigma^{\leq g(n)}$, and thus $\LL$ has a $g(n)$-ceiling as well. Hence, for instance, any polynomial ceiling family of promise problems has an exponential ceiling as well.

Given a function $f$ on $\nat$, the notation $\twon/f(n)$ stands for the subclass of $\twon$, consisting of all families of promise problems having $f(n)$-ceilings. Notice that $g\geq f$ implies $\twon/f(n) \subseteq \twon/g(n)$.
For a set $\FF$ of functions on $\nat$, $\twon/\FF$ denotes the union of all $\twon/f(n)$ for any $f\in\FF$. In a similar way, we obtain $\twod/\FF$ from $\twod$.
These terminologies are also applied to $\twofewu$, $\tworeachu$, $\tworeachfew$, and $\tworeachfewu$.

For convenience, we abbreviate the sets of logarithmic functions, polynomials, exponentials, and sub-exponentials, as ``$\log$'', ``$\poly$'', ``$\mathrm{exp}$'', and ``$\mathrm{subexp}$'', respectively.
This subsection will concentrate on the case where $\FF$ is one of those sets of functions.
We start with an easy case of $\FF=\log$.

\begin{proposition}
$\twon/\slog = \twod/\slog$. More strongly, $\twon/\slog \subseteq \oned$.
\end{proposition}

\begin{proof}
Let us prove that $\twon/\slog\subseteq \oned$ because this inclusion implies $\twon/\slog \subseteq \oned/\slog \subseteq \twod/\slog$.
Toward $\twon/\slog \subseteq \oned$, let  $\LL=\{(L_n^{(+)},L_n^{(-)})\}_{n\in\nat}$ be any family of promise problems in $\twon/\slog$. Take a function $\ell(x)$ of the form $a\log x+b$ for two constants $a,b\geq0$.  Notice that $\Sigma_n\subseteq \Sigma^{\leq \ell(n)}$, where $\Sigma_n=L_n^{(+)}\cup L_n^{(-)}$.   and $|\Sigma^{\ell(n)}|=O(n^a)$.
It then follows that $|\Sigma^{\leq \ell(n)}|=n^{O(1)}$.

Fix $n\in\nat$ arbitrarily. We enumerate all elements in $\Sigma^{\leq \ell(n)}$ as $x^{(n)}_0,x^{(n)}_1,x^{(n)}_2,\ldots$
according to the lexicographic order.
We define $Q_n$ to be composed of all strings of the form $\track{x_i^{(n)}}{a_i}$, where (i) $a_i=+1$ if $x^{(n)}_i\in L_n^{(+)}$, (ii)  $a_i=-1$ if $x^{(n)}_i\in L_n^{(-)}$, and (iii) $a_i=0$ otherwise.
Let us design a 1dfa $M_n$ to read an entire input, say, $x$, determine $i$ satisfying $x=x^{(n)}_i$, and accept (resp., reject) $x$ if $a_i=+1$ (resp., $a_i=-1$). The family $\MM=\{M_n\}_{n\in\nat}$ clearly has polynomial size. Since $\MM$ solves $\LL$, $\LL$ belongs to $\oned$.
\end{proof}

Next, we look into the case of polynomial ceiling. As our main result of this work, we prove the collapse of $\twon/\poly$ down to $\twou/\poly$. Our proof is motivated by a simulation technique of \cite{RA00}.

\begin{theorem}\label{twou-eq-twon}
$\twou/\poly = \twofewu/\poly = \twofew/\poly = \twon/\poly$.
\end{theorem}

We stress that Theorem \ref{twou-eq-twon} does not seem to imply $\twon/\mathrm{subexp} = \twou/\mathrm{subexp}$, $\twon/\sexp =\twou/\sexp$, and $\twon = \twou$. Hence, those equalities are still unknown to hold.

For two other nonuniform state complexity classes, $\tworeachfewu$ and $\tworeachfew$, we obtain the following relationships, depicted in Figure \ref{fig:class-hierarchy}.

\begin{corollary}
$\twod/\poly \subseteq \tworeachu/\poly \subseteq \tworeachfewu/\poly \subseteq \tworeachfew/\poly \subseteq \twou/\poly$.
\end{corollary}

In what follows, we intend to provide the proof of Theorem \ref{twou-eq-twon}.

\vs{-2}
\begin{proofof}{Theorem \ref{twou-eq-twon}}
Because $\twou/\poly \subseteq \twofewu/\poly \subseteq \twofew/\poly \subseteq \twon/\poly$, it suffices to prove the inclusion $\twon/\poly \subseteq \twou/\poly$, which is equivalent to $\twon/\poly \subseteq \twou$. Hereafter, we intend to prove that $\twon/\poly \subseteq \twou$.

It is shown in \cite{Kap14} (re-proven in \cite{Yam22} by a different argument) that $\twon/\poly \subseteq \twod$ iff $\nl\subseteq \dl/\poly$.
In a similar vein, we claim the following.

\begin{yclaim}\label{NL-UL-to-2N-2U}
$\nl\subseteq \ul/\poly$ implies $\twon/\poly \subseteq \twou$.
\end{yclaim}

Assume that Claim \ref{NL-UL-to-2N-2U} is true. Since Reinhardt and Allender \cite{RA00} proved the inclusion $\nl\subseteq \ul/\poly$,
we instantly obtain $\twon/\poly \subseteq \twou$.
Therefore, we aim at proving Claim \ref{NL-UL-to-2N-2U} by adopting the proof technique of \cite{Yam19a,Yam21,Yam22} based on parameterized decision problems. Yamakami \cite{Yam19a} first discovered a close connection between polynomial-size finite automata families and parameterized decision problems. Now, let us introduce the necessary notion of parameterized decision problems.

A \emph{parameterized decision problem} over alphabet $\Sigma$ is a pair $(L,m)$ of a language $L$ over $\Sigma$ and a size parameter $m$, which is a function from $\Sigma^*$ to $\nat$.
A \emph{log-space size parameter} $m$ must satisfy the additional condition that there exists a log-space deterministic Turing machine (DTM) $M$ for which $M$ on any input $x$ produces $1^{m(x)}$ on a write-once\footnote{An output tape is \emph{write once} if (1) its tape head never moves back to the left and (2) whenever the tape head writes a non-blank symbol onto the put tape, it must move to the right next blank cell.} output tape.
An \emph{advice function} $h$ is a function from $\nat$ to $\Gamma^*$ for a certain alphabet $\Gamma$ and it is said to be \emph{polynomially bounded} if $|h(n)|\leq p(n)$ holds for all $n\in\nat$ for an appropriately chosen polynomial $p$. A parameterized decision problem $(L,m)$ belongs to $\PHSP$ if $m$ is \emph{polynomially honest} (that is, there exists a polynomial $q$ such that $|x|\leq q(m(x))$ for all strings $x$). The notation $\para\nl/\poly$ denotes the class composed of all parameterized decision problems $(L,m)$ with log-space size parameters $m$, such that, for each $(L,m)$, there exists a nondeterministic Turing machine (NTM) $M$ equipped with read-only input and advice tapes and a polynomially-bounded advice function $h$ for which $M$ solves $L$ in time $m(x)^{O(1)}$ using space $O(\log{m(x)})$ with free access to advice strings $h(|x|)$ written on the advice tape, where $x$ is a ``symbolic'' input. Similarly, $\para\ul/\poly$ is defined using ``unambiguous'' NTMs (or UTM, for short) instead of NTMs. The reader refers to \cite{Yam17a,Yam19a,Yam22} for more information on parameterized decision problems.

Claim \ref{NL-UL-to-2N-2U} follows immediately from two assertions in the following claim. Given a language $L$ over alphabet $\Sigma$, its complement $\Sigma^*-L$ is succinctly denoted $\overline{L}$.

\begin{yclaim}\label{NL-to-para-NL}
\renewcommand{\labelitemi}{$\circ$}
\begin{enumerate}
  \setlength{\topsep}{-2mm}%
  \setlength{\itemsep}{1mm}%
  \setlength{\parskip}{0cm}%

\item $\nl \subseteq \ul/\poly$ implies $\para\nl/\poly \cap \PHSP \subseteq \para\ul/\poly$.

\item $\para\nl/\poly \cap \PHSP \subseteq \para\ul/\poly$ implies $\twon/\poly \subseteq \twou$.
\end{enumerate}
\end{yclaim}

To close the proof of Theorem \ref{twou-eq-twon}, in what follows, we wish to prove Claim \ref{NL-to-para-NL}.


\ms
(1) Let us begin with the proof of Claim \ref{NL-to-para-NL}(1). Assume that $\nl\subseteq \ul/\poly$. Note that $\nl\subseteq \up/\poly$ implies $\nl/\poly\subseteq \ul/\poly$.
Let us consider an arbitrary parameterized decision problem
$(L,m)$ in $\para\nl/\poly\cap\PHSP$ with a log-space size parameter $m:\Sigma^*\to\nat$ and a language $L$ over alphabet $\Sigma$. Take a polynomially-bounded advice function $h$ and an NTM $M_0$ such that $M_0$ recognizes $\{(x,h(|x|))\mid x\in L\}$ in time
$(m(x))^{O(1)}$ and space $O(\log{m(x)})$.
We wish to verify that $(L,m)\in\para\ul/\poly$.
For this purpose, we define $L_n^{(+)}= L\cap \Sigma_n$ and $L_n^{(-)}= \overline{L}\cap \Sigma_n$, where $\Sigma_n =\{x\in\Sigma^*\mid m(x)=n\}$. We further set $\LL=\{(L_n^{(+)},L_n^{(-)})\}_{n\in\nat}$.

Since $m$ is polynomially honest, we take a polynomial $q$ such that $|x|\leq q(m(x))$ for all $x$. We write $K'$ for the set $\{(x,1^t)\mid x\in L, t\in\nat, m(x)\leq t\}$ and claim that $K'\in\nl/\poly$.
Consider the following algorithm.
On input $(x,1^t)$, we check if $m(x)\leq t$ using log space.
This is possible because $m$ is log-space computable. If $m(x)> t$, then we reject the input; otherwise, we simulate $M_0$ on $(x,h(|x|))$. Clearly, this algorithm recognizes $K'$. This algorithm can be realized by an appropriate NTM running in time $(|x|t)^{O(1)}$ and space $O(\log{|x|t})$ with the help of $h$. Since $|(x,1^t)|=O(|x|+t)$, $K'$ belongs to $\nl/\poly$.

Since $\nl/\poly \subseteq \ul/\poly$ by our assumption, $K'$ must be in $\ul/\poly$. Take a UTM $M$, an advice function $g$, a logarithmic function $\ell'$ such that $M$ recognizes $K'$ using at most $\ell'(|z|)$ space with an access to $g(|z|)$, where $z$  indicates a ``symbolic'' input. We then design a new NTM $N$ for $(K,m')$ as follows.
Define $g'(|x|)=g(|(x,1^t)|)$ for all $x\in\Sigma^*$. On input $x$, compute $n=m(x)$, and run $M$ on $(x,1^n)$ with $g'(|x|)$. Note that $N$'s space usage is $O(\ell'(|x|+n)+\log|x|) \subseteq O(\log(q(m(x))+n)+\log{|x|}) \subseteq O(\log{m(x)})$ since $q$ is a polynomial satisfying $|x|\leq q(m(x))$.
This shows that $(L,m)\in\para\ul/\poly$.

By the definition of $L$, $\LL$ falls into $\twon$. Let $M_n = (Q_n,\Sigma,\{\rhd,\lhd\},\delta_n, q_{0,n},  Q_{acc,n},Q_{rej,n})$. Consider a computation graph $G_n = (V_n,E_n)$ of $M_n$ on each input of length at most $f(n)$.

A configuration of $M_n$ is of the form $(q,l)$, which indicates that $M_n$ is in state $q$ and its tape head scans cell $l$. We view an $(f(n)+2)$-partite graph of the following form: $V_n= Q_n\times [0,f(n)+1]_{\integer}$ and $((q,l),(q',l'))\in E_n$ iff $(q',d)\in \delta_n(q,x_{(l)})$ and $l'=l+d$.


\ms
(2) Next, we wish to prove Claim \ref{NL-to-para-NL}(2). For its proof, we loosely follow an argument made in the proof of \cite[Proposition 5.1]{Yam22}. A key to this argument is Claim \ref{tie-parameter-nonuniform} that ties between parameterized decision problems
and families of promise problems. To describe the claim, we need to recall the following terminology used in \cite{Yam21,Yam22}.

Let $(L,m)$ denote a parameterized decision problem over alphabet $\Sigma$ and let $\LL=\{(L_n^{(+)},L_n^{(-)})\}_{n\in\nat}$ be any family of promise problems. We say that $\LL$ is \emph{induced from} $(L,m)$ if, for any index $n\in\nat$, $L_n^{(+)}=L\cap \Sigma_n$ and $L_n^{(-)}=\overline{L}\cap \Sigma_n$, where $\Sigma_n$ denotes the set $\{x\in\Sigma^*\mid m(x)=n\}$.
For two families of promise problems  $\LL=\{(L_n^{(+)},L_n^{(-)})\}_{n\in\nat}$ and $\KK=\{(\KK_n^{(+)},\KK_n^{(-)})\}_{n\in\nat}$, we say that $\KK$ is an \emph{extension} of $\LL$ if, for any index $n\in\nat$, $L_n^{(+)}\subseteq K_n^{(+)}$ and $L_n^{(-)}\subseteq K_n^{(-)}$. We say that $\LL$ is \emph{$\dl$-good} if the set $\{1^n\# x\mid n\in\nat, x\in L_n^{(+)}\cap L_n^{(-)}\}$ belongs to $\dl$. A collection $\FF$ of families of promise problems is \emph{$\dl$-good} if all elements of $\FF$ has an $\dl$-good extension in $\FF$.
From a family $\LL=\{(L_n^{(+)},L_n^{(-)})\}_{n\in\nat}$, we introduce $K_n^{(+)} = \{1^n\# x\mid x\in L_n^{(+)}\}$ and $K_n^{(-)}=\{1^n\# x\mid x\notin L_n^{(+)}\} \cup S_n$, where $\Sigma_{\#}=\Sigma\cup\{\#\}$ and $S_n = \{z\# x\mid z\in \Sigma^n-\{1^n\},x\in (\Sigma_{\#})^*\}\cup \Sigma^n$.  We say that $(K,m)$ is \emph{induced from} $\LL$ if $K$ equals $\bigcup_{n\in\nat} K_n^{(+)}$ and $m:(\Sigma_{\#})^*\to\nat$ satisfies the following: for any string $w$ of the form $1^n\# x$ with $x\in L_n^{(+)}\cup L_n^{(-)}$, $m(w)=n$, and for all other strings $w$, $m(w)=|w|$. Note that $\overline{K}$ equals $\bigcup_{n\in\nat} K_n^{(-)}$.
See \cite{Yam22} for more information.

The $\dl$-goodness of $\twon/\poly$ was  shown in \cite{Yam22}. It is, however, unclear that $\twou/\poly$ is also $\dl$-good. We then describe a key statement, Claim \ref{tie-parameter-nonuniform}.

\begin{yclaim}\label{tie-parameter-nonuniform}
Let $(\CC,\DD)\in\{(\nl,\twon), (\ul,\twou)\}$. Let $\LL=\{(L_n^{(+)},L_n^{(-)})\}_{n\in\nat}$ denote any family of promise problems and let $(L,m)$ and $(K,m)$ be any two parameterized decision problems with the common logspace size parameter $m$. (a) If $\LL$ is induced from $(L,m)$, then $(L,m)\in\para\CC/\poly\cap \PHSP$ iff $\LL\in\DD/\poly$. (b) If $\LL$ is $\dl$-good and $(K,m)$ is induced from $\LL$, then $(K,m)\in\para\CC/\poly\cap \PHSP$ iff $\LL\in\DD/\poly$.
\end{yclaim}

\begin{proof}
The case of $\CC=\nl$ and $\DD=\twon$ was proven in \cite{Yam22}. Now, let us consider the case of $\CC=\ul$ and $\DD=\twou$. The proof idea is similar to the first case of $\CC=\nl$ and $\DD=\twon$.

(a) Let $\LL=\{(L_n^{(+)},L_n^{(-)})\}_{n\in\nat}$ denote any family of promise problems induced from $(L,m)$. Assume that $(L,m)\in\para\ul/\poly \cap \PHSP$. There are an advice function $h$ of output size $m(x)^{O(1)}$ and a UTM $M$ that solves $(L,m)$ in time $m(x)^{O(1)}$ using space $O(\log{m(x)})$.
Since $m$ is a logspace size parameter, it is polynomially bounded. Take a suitable polynomial $q$ satisfying $m(x)\leq q(|x|)$ for all $x$. The space usage of $M$ is thus upper-bounded by $O(\log|x|)$. Note that $\Sigma_n = \{x\mid m(x)=n\}\subseteq \Sigma^{\leq q(n)}$. We also assume that $|h(i)|\leq p(m(x))$.
Hereafter, we intend to verify that  $\LL\in\twou/\poly$.

We then construct a 1nfa $N_n$ that takes $x$ as an input and simulates $M$ on input $(x,h(|x|))$. For this purpose, we embed the information on $(h_{i1},h_{i2},\ldots,h_{ik_i})$ for $h(i)=h_{i1}h_{i2}\cdots h_{i|h(i)|}$ into $|h(i)|$ inner states of $N_n$ so that it can recover $h(i)$ from these inner states. It follows that $N_n$ is unambiguous on $\Sigma_n$. Since $\{N_n\}_{n\in\nat}$ solves $\LL$, we conclude that $\LL\in\twou/\poly$.

Next, we assume that $\LL\in\twou/\poly$. Take a family $\NN=\{N_n\}_{n\in\nat}$ of polynomial-size 1ufa's solving $\LL$. We then define an NTM $M$ that recognizes $(L,m)$ with the help of an advice function $h$. We encode the transitions of $N_n$ into an advice string. On input $x$, compute $n=m(x)$, and simulate $N_n$ on $x$. It then follows that $M$ is unambiguous and recognizes $(L,m)$.

(b) Assume that $\LL$ is $\dl$-good and that $(K,m)$ is induced from $\LL$. If $(K,m)\in\para\ul/\poly\cap\PHSP$, then we take an advice function $h$ of output size $m(x)^{O(1)}$ and a UTM $M$ for which $M$ solves $(K,m)$ in time $m(x)^{O(1)}$ using space $O(\log{m(x)})$ with the help of $h$. We construct a 1nfa $N_n$ that behaves as follows. Recall from the proof of (a) the embedding of $h(i)$ into $|h(i)|$ inner states.
On input $x$, recover $h(|x|)$ from $N_n$'s inner states, and simulate $M$ on $(1^n\# x,h(|x|))$. This implies that $N_n$ solves $(L_n^{(+)},L_n^{(-)})$.

On the contrary, assume that $\LL\in\twou/\poly$ and take a family $\{N_n\}_{n\in\nat}$ of polynomial-size 1ufa's solving $\LL$. Since $N_n$ is of polynomial size, we can encode $N_n$ into an advice string $\alpha_n$ of polynomial length.
Since $\LL$ has a polynomial ceiling, take a constant $k\geq1$ satisfying $L_n^{(+)}\cup L_n^{(-)}\subseteq \Sigma^{\leq n^k+k}$ for any $n\in\nat$. We then construct an NTM $M$ such that, on input $y$, check if $y$ has the form $1^n\# x$ with $|x|\leq n^{k}+k$. If so, simulate $N_n$ on $x$.
Since $m(y)=n$, it follows that $|y|=n+|x|+1 \leq n^{k+1}+k\leq m(y)^{k+1}+k$.
Otherwise, reject $x$ immediately.
This machine $M$ is unambiguous and also recognizes $K$ in time polynomial in $m(y)$ using space logarithmic in $m(y)$.
\end{proof}


Assume that $\para\nl/\poly\cap \PHSP \subseteq \para\ul/\poly$.
Let $\LL=\{(L_n^{(+)},L_n^{(-)})\}_{n\in\nat}$ be an arbitrary element of $\twon/\poly$ and take two polynomials $p,q$ and a family $\MM=\{M_n\}_{n\in\nat}$ of 2nfa's that solves $\LL$ with the following properties: each $M_n$ has at most $p(n)$ inner states and $\Sigma_n$ ($=L_n^{(+)}\cup L_n^{(-)}$) is included in $\Sigma^{\leq q(n)}$.
Hereafter, we wish to prove that $\LL$ belongs to $\twou$.

For each index $n\in\nat$, we define $K_n^{(+)} = \{1^n\# x\mid x\in L_n^{(+)}\}$ and $K_n^{(-)} = \{1^n\# x\mid x\in L_n^{(-)}\} \cup \{z\# x\mid z\in\Sigma^n-\{1^n\}, x\in\Sigma^*_{\#}\} \cup \{z\mid z\in\Sigma^n\}$, where $\Sigma_{\#} = \Sigma\cup\{\#\}$. We set $K=\bigcup_{n\in\nat} K_n^{(+)}$ and $K^c=\bigcup_{n\in\nat} K_n^{(-)}$. It then follows that $K\cup K^c=\Sigma^*_{\#}$ and $K\cap K^c=\setempty$; thus, $K^c$ coincides with the complement $\overline{K}$ of $K$.
We define $m'(w)=n$ if $w=1^n\# x$ for a certain $x\in L_n^{(+)}\cup L_n^{(-)}$ and $m'(w)=|w|$ otherwise.
Note that $m'$ is log-space computable and also polynomially honest.
Since $(L_n^{(+)},L_n^{(-)})$ is solved by $M_n$ for each $n\in\nat$,
$K$ is solvable nondeterministically in time $(|x|m(x))^{O(1)}$ and space $O(\log{m(x)})$ with the use of an appropriate polynomial-size advice function.
Thus, $(K,m')$ belongs to $\para\nl/\poly \cap \PHSP$.

The assumption $\para\nl/\poly \cap \PHSP \subseteq \para\ul/\poly$ makes  $(K,m')$ fall into $\para\ul/\poly$.
Take an advice function $h$ and a UTM $N$ that solve $K$ in time $(|x|m'(x))^{O(1)}$ and space $O(\log{m'(x)})$.
Consider the algorithm that, on input $x$ with index $n\in\nat$, generate both $1^n\# x$ and $h(|x|)$ and then run $N$ on $(1^n\#x, h(|x|))$ to produce an outcome. An appropriate unambiguous 2nfa, say, $N'_n$ can realize this algorithm since we can store $h(|x|)$ as a series of inner states.
We thus conclude that $N'_n$ solves $(L_n^{(+)},L_n^{(-)})$.  Therefore, $\LL$ belongs to $\twou$.

This completes the proof of Theorem \ref{twou-eq-twon}.
\end{proofof}

\subsection{Case of No Ceiling Restriction}

We have discussed the case of bounded ceilings in Section \ref{sec:bounded-ceilings}. We next turn our attention to the case of no ceiling restriction.
Kapoutsis \cite{Kap14} earlier demonstrated that $\twon\subseteq \twod$ holds exactly when $\twon/\supexp \subseteq \twod$, where ``$\supexp$'' denotes the set of all super exponentials on $\nat$. Similarly, we can assert that, even in our unambiguity/fewness setting, the case of no ceiling restriction on families of promise problems is logically equivalent to the case of $\supexp$-ceiling restriction on those families. More formally, we wish to prove the following.

\begin{theorem}\label{super-expo-bound}
Let $\AAA$ denote the set $\{\mathrm{D}, \mathrm{ReachU}, \mathrm{ReachFewU}, \mathrm{ReachFew}, \mathrm{U}, \mathrm{FewU},\mathrm{Few}, \mathrm{N}\}$.
For any two classes $\CC,\DD$ taken from $\AAA$, it follows that $2\CC \subseteq 2\DD$ iff $2\CC/\supexp \subseteq 2\DD$.
\end{theorem}

The proof of this theorem can be obtained by slightly modifying the proof in \cite{Kap14}. For the sake of completeness, nevertheless, we include the proof of the theorem.

For a finite automaton $M$, the notation $L(M)$ means the set of all strings accepted by $M$.  In the following proof, for simplicity, we assume without loss of generality that 2nfa's make their tape head return to the start cell when they halt.

\begin{proofof}{Theorem \ref{super-expo-bound}}
Let $\CC,\DD$ be any classes in $\AAA$. Since $2\CC/\supexp \subseteq 2\CC$, it is obvious that $2\CC \subseteq 2\DD$ implies $2\CC/\supexp \subseteq 2\DD$.  In what follows, we intend to prove the converse.
Assuming that $2\CC/\supexp \subseteq 2\DD$, we aim at proving that $2\CC\subseteq 2\DD$.
Let $f(n)$ denote any function in $\supexp$. It then follows that, for any polynomial $q$,  $f(n)>2^{q(n)}$ holds for all but finitely many numbers $n\in\nat$.
Let $\LL=\{(L_n^{(+)},L_n^{(-)})\}_{n\in\nat}$ denote any family of promise problems in $2\CC$ and take a family $\NN=\{N_n\}_{n\in\nat}$ of polynomial-size 2nfa's that solves $\LL$, where every $N_n$ must satisfy the condition imposed by the definition of $2\CC$. For each $n\in\nat$, let $N_n=( Q'_n,\Sigma, \{\rhd,\lhd\}, \delta'_n, q'_{0,n}, Q'_{acc,n},Q'_{rej,n})$. There exists a polynomial $q$ satisfying $|Q'_n|\leq q(n)$ for all $n\in\nat$.

Fix $n$ arbitrarily. We expand $(L_n^{(+)},L_n^{(-)})$ to the ``language'' $L(N_n)$ induced by $N_n$. Notice that $L_n^{(+)}\subseteq L(N_n)$ and $L_n^{(-)}\subseteq \overline{L(N_n)}$. For convenience, we write $N_n(x)$ to denote the outcome (i.e., acceptance or rejection) of $N_n$ on input $x$.
From this language $L(N_n)$, we define another promise problem $(K_n^{(+)},K_n^{(-)})$ by setting $K_n^{(+)}=\{x\in L(N_n) \mid |x|\leq f(n)\}$ and $K_n^{(-)}=\{x\in \overline{L(N_n)} \mid |x|\leq f(n)\}$.
Let us consider the family $\KK=\{(K_n^{(+)},K_n^{(-)})\}_{n\in\nat}$. Since $\KK$ has an $f(n)$-ceiling, $\KK$ must be in $2\CC/f(n)$, which is further included in $2\CC/\supexp$ since $f\in\supexp$.

By our assumption, $\KK$ belongs to $2\DD$. Take a family  $\MM=\{M_n\}_{n\in\nat}$ of polynomial-size 2nfa's with $M_n =(Q_n,\Sigma,\{\rhd,\lhd\}, \delta_n,q_{0,n}, Q_{acc,n},Q_{rej,n})$ that solves $\KK$, where $|Q_n|\leq p(n)$ holds for a fixed polynomial $p$ independent of $n$ and each $M_n$ satisfies
the condition imposed for $2\DD$.
Remember that $M_n$ may take strings of arbitrary lengths as its inputs.
However, for any input $x$ of length at most $f(n)$, $M_n(x)$ coincides with $N_n(x)$ by the definition of $\KK$.
We then compare between the behaviors of $M_n$ and $N_n$. Define $r(n)=2(p(n)^2+q(n)^2)$ for all $n\in\nat$, which is a polynomial satisfying $2(|Q_n|^2+|Q'_n|^2)\leq r(n)$.
We then claim the following statement concerning the ``discrepancy'' set $A=\{n\in\nat\mid L(M_n)\neq L(N_n) \}$.

\begin{yclaim}\label{size-of-A}
$A \subseteq \{n\in\nat\mid 2^{r(n)}\geq f(n)\}$.
\end{yclaim}

\begin{proof}
We first show that, if $L(M_n)\neq L(N_n)$, then there exists a string $y$ of length at most $2^{r(n)}$ for which $M_n(y)\neq N_n(y)$. If this is true, then it follows from $M_n(y)\neq N_n(y)$ that $|y|>f(n)$.
We then obtain $2^{r(n)}\geq f(n)$. In conclusion, $A$ is a subset of $\{n\in\nat\mid 2^{r(n)}\geq f(n)\}$.

Fix an arbitrary index $n\in A$ and take the shortest string $x$ satisfying $M_n(x)\neq N_n(x)$. Recall from Section \ref{sec:FA-CS} the notion of critical computation path. From $(M_n, N_n, x)$, we obtain two critical computation paths $\gamma_{M_n}(x)$ and $\gamma_{N_n}(x)$ associated with $M_n$ and $N_n$, respectively.
By Lemma \ref{upper-bound-length}, the crossing sequence at any boundary of $\gamma_{M_n}(x)$ (resp., $\gamma_{N_n}(x)$) has length at most $2|Q_n|$ (resp., $2|Q'_n|$).

We define $D$ to be the collection of all pairs $(\alpha,\alpha')$ such that $\alpha$ (resp., $\alpha'$) is a crossing sequence of length at most $2|Q_n|$ (resp., $2|Q'_n|$) of $\gamma_{M_n}(x)$ (resp., $\gamma_{N_n}(x)$) at the same boundary. Note that $|D|\leq |Q_n|^{2|Q_n|}\cdot |Q'_n|^{2|Q'_n|}\leq  2^{2p(n)^2+2q(n)^2}= 2^{r(n)}$. Thus, if $|x|\geq 2^{r(n)}+1$, then there exist a distinct pair $(s_1,s_2)$ and another pair $(\alpha,\alpha')$ for which $\gamma_{M_n}(x)$ and $\gamma_{N_n}(x)$ have a crossing sequence pair $(\alpha,\alpha')$ in $D$ at the both boundaries $s_1$ and $s_2$.
Therefore, given any input $x$ with $|x|\geq 2^{r(n)}+1$, there exists another input $y$ of length at most $2^{r(n)}$ for which $M_n(x)=M_n(y)$ and $N_n(x)=N_n(y)$. From this, we obtain $M_n(y)\neq N_n(y)$. Since $|y|<|x|$, this contradicts the minimality of $x$.
\end{proof}

Finally, we modify $M_n$ into $M'_n$ in order to solve $(L_n^{(+)},L_n^{(-)})$. By the choice of $f$, $\{n\in\nat\mid 2^{r(n)}\geq f(n)\}$ is a finite set, and thus $A$ is also a finite set. Let $c$ denote the largest number in $A$.
Let $n$ be any number in $\nat$. If $n\in A$, then we take a 1dfa that exactly simulates $N_n$ with $2^{O(|Q'_n|)}$ inner states and we set $M'_n$ to be this 1dfa. It is important to remark that, since $\oned\subseteq 2\DD$, $M'_n$ is a 2nfa satisfying the condition imposed for $2\DD$. Otherwise, we define $M'_n$ to be exactly $M_n$. Since $L(M_n)=L(N_n)$ for all $n>c$, the state complexity of $M'_n$ is upper-bounded by a certain constant, independent of $n$.
Therefore, $M'_n$ correctly solves $(L_n^{(+)},L_n^{(-)})$, as requested.
\end{proofof}

\section{A Short Discussion and Open Problems}\label{sec:discussion}

The \emph{theory of nonuniform polynomial state complexity} dates back to the late 1970s \cite{BL77,SS78} and it seems to have been left unattended  until the late 2000s in the field of automata and computational complexity. This intriguing theory was rediscovered in 2009 \cite{Kap09} and it has recently gained its pace with a series of studies \cite{Gef12,Kap12,Kap14,KP15, Yam19a,Yam19d,Yam21,Yam22,Yam23b}.

In this work, we have attempted to continue the investigation of the computational complexity of families of various types of finite automata of polynomial state complexity.
In particular, we have concentrated on nonuniform families of polynomial-size nondeterministic finite automata with unambiguity/fewness conditions imposed on their computation graphs and accepting computation paths. Inspired by  \cite{PTV12}, we have studied in this work six complexity classes of  families of promise problems solvable by those specific finite automaton families.
When tape heads of underlying finite automata are limited to move in only one direction, we have shown mutual class separations, namely, most of those complexity classes are distinct from each other. On the contrary, when the tape heads are allowed to move in both directions, four of the six complexity classes have been shown to collapse.
All those results have been illustratively summarized in Figure \ref{fig:class-hierarchy}.

As for problems left untold in this work, we wish to list four relevant topics for the interested reader.

\renewcommand{\labelitemi}{$\circ$}
\begin{enumerate}
  \setlength{\topsep}{-2mm}%
  \setlength{\itemsep}{1mm}%
  \setlength{\parskip}{0cm}%

\item It is important to point out that Figure \ref{fig:class-hierarchy} is not yet complete. We therefore need to complete this figure by filling its missing parts. For instance, is it true that $\oneu\nsubseteq \onereachfewu$ or even $\oneu\nsubseteq \onereachfew$? Does $\onereachfew\nsubseteq \onefewu$ hold?

\item We have shown that $\twon/\poly$ coincides with $\twou/\poly$, but this result does not seem to extend to the collapse of $\twon/\poly$ down to $\tworeachfew/\poly$ (or $\tworeachfewu/\poly$ or even $\tworeachu/\poly$). Does such a collapse actually occur?

\item The choice of ceiling (e.g., polynomial and exponential) may affect the computational strengths of nonuniform state complexity classes. When we replace ``polynomial ceilings'' by ``exponential ceilings'', for instance, is it true that $\twon/\sexp\subseteq \twou$? Is it also true that  $\tworeachfew/\sexp\subseteq \tworeachfewu$?

\item Recently, as a natural extension of finite automata, nonuniform families of pushdown automata were studied in \cite{Yam21}.
    We write $\twodpd$ for the collection of all families of promise problems solvable by nonuniform families of 2-way pushdown automata having polynomial stack-state complexity. For more precise definition, refer to \cite{Yam21}. It is unknown that $\twon\subseteq \twodpd$. For other low complexity classes, such as $\tworeachfewu$ and $\tworeachfew$, is it true that  $\tworeachfewu\subseteq \twodpd$ or  even $\tworeachfew\subseteq \twodpd$?

\item When a tape head of a one-way finite automaton is further allowed to stay still at any moment (that is, makes $\varepsilon$-moves), we empathetically call an associated machine a \emph{1.5-way finite automaton}. It is shown in \cite{Yam22} that $\oned=\mathrm{1.5D}$ but $\onebq\neq \mathrm{1.5BQ}$, where $\onebq$ and $\mathrm{1.5BQ}$ are defined by bounded-error 1-way and 1.5-way quantum finite automata having polynomially many inner states. It seems interesting to ask types of finite automata make the 1.5-way head move different from the 1-way head move.
\end{enumerate}


\let\oldbibliography\thebibliography
\renewcommand{\thebibliography}[1]{%
  \oldbibliography{#1}%
  \setlength{\itemsep}{-2pt}%
}
\bibliographystyle{alpha}

\end{document}